\documentclass[journal,10pt,final]{IEEEtran}
\IEEEoverridecommandlockouts

\usepackage{cite}
\usepackage{indentfirst}
\usepackage{graphicx}
\usepackage{changepage}
\usepackage{stfloats}
\usepackage{amsfonts,amssymb}
\usepackage{bm}
\usepackage{algorithm}
\usepackage{algorithmic}
\usepackage{amsmath}
\usepackage{amsmath,amssymb,amsfonts}
\usepackage{times}
\usepackage{url}
\usepackage{textcomp}
\usepackage{xcolor}
\usepackage{color}
\usepackage{setspace}
\usepackage{enumitem}
\usepackage{float}
\usepackage{diagbox}
\usepackage[T1]{fontenc}
\usepackage[utf8]{inputenc}
\usepackage{authblk}
\usepackage{threeparttable}
\usepackage{booktabs}
\usepackage{footnote}
\usepackage{graphicx}
\usepackage{pifont}
\usepackage{subfigure}
\usepackage{mathrsfs}

\renewcommand{\theenumi}{\alph{enumi}}

\newtheorem{proposition}{Proposition}

\newenvironment{proof}{{\indent \indent \it Proof:\quad}}{\hfill $\blacksquare$\par}
\newtheorem{remark}{Remark}

\begin{document}

\title{Low-Altitude UAV Tracking via Sensing-Assisted Predictive Beamforming}

\author{\IEEEauthorblockN{Yifan~Jiang, Qingqing~Wu, Hongxun~Hui, Wen~Chen, and Derrick~Wing~Kwan~Ng}

\thanks{This work is supported in part by NSFC 62371289, in part by Shanghai Jiao Tong University 2030 Initiative, and in part by the Science and Technology Development Fund, Macao SAR (File no. 0050/2025/AIJ, File no. 001/2024/SKL, and 0002/2025/EQP).}
\thanks{Yifan~Jiang is with Shanghai Jiao Tong University, Shanghai 200240, China, and also with the State Key Laboratory of Internet of Things for Smart City, University of Macau, Macao 999078, China (email: yc27495@umac.mo). 
Qingqing~Wu and Wen~Chen are with the Department of Electronic Engineering, Shanghai Jiao Tong University, Shanghai 200240, China (e-mail: \{qingqingwu@sjtu.edu.cn; wenchen@sjtu.edu.cn\}).
Hongxun~Hui is with the State Key Laboratory of
Internet of Things for Smart City and Department of Electrical and Computer
Engineering, University of Macau, Macao, 999078 China (email: hongxunhui@um.edu.mo).
Derrick~Wing~Kwan~Ng is with the School of Electrical Engineering and Telecommunications, University of New South Wales, Sydney, NSW 2052, Australia (e-mail: w.k.ng@unsw.edu.au).
}}

\maketitle

\begin{abstract} 
Sensing-assisted predictive beamforming shows significant promise for enhancing various future unmanned aerial vehicle (UAV) applications in integrated sensing and communication (ISAC) systems. 
However, the impact of such beamforming technique on the communication reliability was largely unexplored and challenging to characterize.
To fill this research gap and tackle this issue, this paper proposes a cellular-connected UAV tracking scheme leveraging extended Kalman filtering (EKF), where the predicted UAV trajectory, sensing duration ratio, and target constant received signal-to-noise ratio (SNR) are jointly optimized to maximize the outage capacity at each time slot. 
To address the implicit nature of the objective function, analytical outage probability (OP) approximations are proposed based on second-order Taylor expansions, providing an efficient and full characterization of outage capacity. 
Subsequently, an efficient algorithm is proposed based on a combination of bisection search and successive convex approximation (SCA) to address the non-convex optimization problem with guaranteed convergence. 
To further reduce computational complexity, a second efficient algorithm is developed based on alternating optimization (AO). 
Simulation results validate the accuracy of the derived OP approximations, the effectiveness of the proposed algorithms, and the significant outage capacity enhancement over various benchmarks. 
Furthermore, we show that the optimized predicted UAV trajectory tends to be parallel to the base station's uniform linear array antennas with a nonzero minimum distance, indicating a trade-off between decreasing path loss and enjoying wide beam coverage for outage capacity maximization.
\end{abstract}


\section{Introduction}

The production of unmanned aerial vehicles (UAVs), also known as drones, is anticipated to experience sustained rapid growth over the next decade, reaching a global market value of over 70 billion dollars by 2030 \cite{UAV1}. 
To date, the emerging low-altitude economy (LAE) has attracted significant worldwide attention, which unprecedently utilizes vertical space below an altitude of 1000 meters for numerious applications such as logistics, industrial monitoring, emergency rescue, and air taxis \cite{LAE1, LAE3, LAE4}.
However, it can be envisioned that signal interference and network congestion will intensify considerably due to the explosive increase in UAV equipment.
Under these challenging circumstances, it is crucial to guarantee the communication and tracking performance of UAV users and targets, since these metrics serve the foundations for the aforementioned applications. 
In addition, despite the independent progress made in UAV communication and tracking, such as short packet communication and global navigation satellite system \cite{qqw2021JSAC}, incompatibility among standalone systems designed for different functions causes inefficient use of hardware and spectrum resources, which degrades overall system performance.

In recent years, integrated sensing and communication (ISAC) has been proposed and widely investigated as an enabling technology for the upcoming sixth-generation (6G) network \cite{ISAC2,fwd,ZhenDu-2,Guangji}. 
By effectively utilizing the inherent reciprocity between sensing and communication, ISAC is expected to offer more precise and reliable wireless coverage for UAVs, thereby mitigating signal interference among multiple UAVs and enabling the access of a massive number of UAVs. 
Meanwhile, hardware and spectrum resources can be efficiently integrated, significantly reducing system overhead and improving overall spectral efficiency. 
Therefore, ISAC technologies present a vital solution to address the aforementioned issues, while simultaneously providing high-quality communication and tracking services for UAVs.

Among various existing signal processing and architecture designs for ISAC, sensing-assisted predictive beamforming is appealing due to its effectiveness in enhancing both target tracking accuracy and communication links for users concurrently \cite{ZhenDu-2,PB1,Relia1}. 
Specifically, sensing-assisted predictive beamforming refers to the design of beamforming vectors based on both predicted and measured user information, typically user directions or positions. 
By utilizing the predicted user information, the system overhead associated with conventional beam training and feedback can be considerably reduced and the beam alignment accuracy can be improved from the perspective of Bayesian filtering \cite{Relia1}. 
Existing studies have demonstrated that sensing-assisted predictive beamforming can significantly enhance the performance of vehicular networks \cite{ZhenDu-2,PB1,PB6,XiaoMeng,PB2,PB3,Relia1,Relia2}.
For instance, a predictive beamforming framework was proposed in \cite{PB1} for vehicular communications, leveraging the popular extended Kalman filtering (EKF) technique and substantially enhancing the achievable rate over conventional feedback-based communication schemes.
In \cite{PB6}, the average achievable sum-rate was maximized by optimizing the predictive beamformer for multi-user vehicular communications, exhibiting an upper-bound-approaching communication performance.
Other works have extended sensing-assisted predictive beamforming designs to a wide range of cases with complex driving behaviour or trajectories \cite{XiaoMeng,PB2}, end-to-end transmission \cite{PB3}, extended targets \cite{Relia2}, etc.

Given the potential of sensing-assisted predictive beamforming, more recent works have explored its applications for system performance enhancement in low-altitude UAV networks \cite{ISAC-UAV-1,PB4,secure1}.
In \cite{ISAC-UAV-1}, a predictive beam management approach utilizing visual information was proposed based on the EKF for multi-UAV communication and tracking, showing a $20.34\%$ data rate gain under a 64-antenna setup. 
In \cite{PB4} and \cite{secure1}, UAV eavesdropper tracking schemes were proposed based on sensing-assisted predictive beamforming to improve the secrecy rates of legitimate users and the adversary UAV tracking accuracy.
The above works are limited to scenarios with unknown or uncontrollable UAV trajectories. 
Nevertheless, unlike vehicular networks with constrained vehicle trajectories, there are many cases in low-altitude UAV networks where UAV trajectories can be optimized and controlled in real time within a large volume of three-dimensional (3D) space, such as cellular-connected drone surveillance, displays, and tourism \cite{qqw2021JSAC,LAE1,LAE3}.
In these cases, the achievable rate and Cramér-Rao bound (CRB) for UAV movement estimation are not only affected by predictive beamforming designs but also highly dependent on UAV trajectories \cite{Khalili-1,Khalili-2,MARichards}. 
Therefore, the joint optimization of UAV trajectory and predictive beamforming is promising for improving the communication and tracking performance of the aforementioned low-altitude UAV applications. 
However, such joint optimization is challenging due to the complicated coupling between the UAV trajectory and BS beamformer.
Moreover, the UAV trajectory can only be partially optimized because of the inherently random aerial environmental variations and control errors in practice \cite{qqw2021JSAC,2020TCOM-Uncertain,ISAC-UAV-2}, thus requiring dedicated optimization approaches to address UAV trajectory uncertainties.

Furthermore, the majority of existing works on sensing-assisted predictive beamforming have focused on spectral efficiency improvement, which is insufficient for characterizing communication reliability.
However, modern low-altitude UAV applications have stringent requirements for communication reliability, such as an extremely low probability of signal outage \cite{qqw2021JSAC,LAE1}.
So far, relatively few investigations have examinated the improvements in beam alignment probability brought about by sensing-assisted predictive beamforming \cite{Relia1,Relia2,XiaoMeng}.
In \cite{Relia1}, the impact of beamwidth on the beam alignment probability was studied for predictive beamforming-enabled vehicle tracking.
Following \cite{Relia1}, dynamic beamwidth designs for vehicle tracking were proposed to improve the tracking or communication performance incorporating the beam alignment probability \cite{Relia2,XiaoMeng}.
Nevertheless, the beam alignment probability cannot characterize the reliable data rate or link capacity, which are crucial performance metrics for practical system designs.
Instead, the outage probability (OP) and outage capacity are more appropriate and fundamental communication performance metrics to be studied \cite{goldsmith}.
Far from solely considering an uninvestigated performance metric, important insights can be drawn by studying the OP and outage capacity into the roles of sensing accuracy on communication reliability and the mechanisms of maximizing the reliable communication performance via predictive beamforming and UAV trajectory optimization.
Meanwhile, it is intractable to derive the OP and outage capacity directly from the studied beam alignment probability owing to the different definitions, making it challenging to characterize the OP and outage capacity in the sensing-assisted predictive beamforming scheme.

Motivated by the aforementioned issues, we investigate the outage capacity characterization and maximization via UAV trajectory optimization in this paper. 
Specifically, a cellular-connected UAV is served and also concurrently tracked via EKF by a monostatic ISAC BS. 
Through remote control from the BS, the predicted UAV trajectory can be proactively controlled, although it is interfered by control noise modeled as a Gaussian random process. 
Within each short time slot, the UAV motion state is assumed to be deterministic yet unknown in advance. 
Meanwhile, the communication performance directly depends on a sensing duration ratio between the prediction and measurement durations at each time slot under the sensing-assisted predictive beamforming scheme. 
As a result, the system communication reliability at each time slot can be evaluated by OPs and outage capacities at the prediction and measurement stages, respectively. 
The main contributions of this paper are summarized as follows: 
\begin{itemize}
    \item A joint UAV tracking and outage capacity maximization scheme is proposed for reliable communication, where an optimization problem for outage capacity maximization is formulated and addressed at each time slot to optimize the predicted UAV trajectory, sensing duration ratio, and target constant received signal-to-noise ratios (SNRs), subject to constraints on UAV velocity and a maximum tolerable OP. 
    \item To address the implicit and non-convex objective function and constraints in the formulated problem, closed-form approximations of OPs for both the prediction and measurement stage are proposed based on second-order Taylor expansions, enabling the full characterization of outage capacity and a more tractable optimization problem formulation. 
    To the best of our knowledge, this paper represents the first effort to characterize the outage capacity under the sensing-assisted predictive beamforming scheme in ISAC systems. 
    \item An efficient algorithm is proposed to handle the formulated optimization problem with guaranteed convergence, in which the formulated problem is decomposed into two feasibility problems addressed by the bisection search and SCA, respectively. Moreover, the updating rules between the two feasibility problems are heuristically designed based on the proved monotonicity of apporximated OPs with respect to (w.r.t.) the target constant received SNRs. 
    To further reduce computational complexity and avoid unnecessary trials involving infeasible solutions, a second efficient algorithm is proposed capitalizing alternating optimization (AO), which maximizes the outage capacity within a few iterations. 
    \item Simulation results validate the effectiveness of our proposed OP approximations, algorithms, and outage capacity maximization scheme. 
    In addition, in the prediction mean square error (MSE)-dominant case, our proposed scheme achieves a significant outage capacity improvement compared to benchmarks. 
    Moreover, our results reveal that the optimized predicted UAV trajectory ends up with being parallel to the BS uniform linear array (ULA) antennas with a nonzero minimum distance, which also demonstrates a trade-off between reducing path loss and enlarging beam coverage area. 
\end{itemize}

\emph{Notation:} $\mathbf{0}_{m}$ and $\mathbf{1}_{m}$ denote a $m\times 1$ column vector with all elements equal to 0 and 1, respectively. 
$\mathcal{O}(\cdot)$ represents the big-O notation for computational complexity. 
$\mathbb{E}_{x}[\cdot]$ is statistical expectation w.r.t. the distribution of $x$. 
$\mathcal{N}(\mathbf{x},\mathbf{Q})$ denotes a real-valued Gaussian distribution with a mean vector $\mathbf{x}$ and covariance matrix $\mathbf{Q}$ and $\sim$ means ``distributed as''. 
$\preceq$ is the element-wise component inequality. 
$\otimes$ is the Kronecker product. 
$\mathrm{diag}(b_{1},...,b_{L})$ denotes a diagonal matrix with $b_{1},...,b_{L}$ being its diagonal elements. 
For an arbitrary matrix $\mathbf{A}$, $\mathbf{A}^{T}$, $\mathbf{A}^{-1}$, $\mathrm{det}(\mathbf{A})$, and $[\mathbf{A}]_{ij}$ denote its transpose, inverse, determinant, and $(i,j)$-th element, respectively. 
$\nabla f(\cdot)$ represents the gradient of the function $f(\cdot)$.
Other key notations are summarized in Table I.

\section{System Model}\label{SecII}

We consider a terrestrial BS that employs downlink ISAC signals to simultaneously track and communicate with a single-antenna cellular-connected UAV.\footnote{The considered scenario can be readily extended to multi-UAV scenarios, where the BS serves multiple UAVs with time-division or frequency-division multiple access schemes.}
As an initial study, it is assumed that the UAV flies at a fixed altitude of $H$ m, and the BS is equipped with ULAs comprising $N_{\text{t}}$ transmit antennas and $N_{\text{r}}$ receive antennas.\footnote{This model can be readily extended to the case with 3D trajectory optimization by incorporating the altitude into the state vector, and the case with planar arrays by additonally modeling the elevation angle measurement and formulating the beamforming gain as a product of antenna gain from two orthogonal directions.}
Furthermore, the uncertainty of the UAV motion state (i.e., the UAV position and velocity) is considered owing to practical issues such as control errors \cite{2020TCOM-Uncertain,ISAC-UAV-2}. 
Moreover, with a sufficiently short time interval $\Delta T$ s, the UAV motion state can be assumed to be invariant \cite{qqwMultiUAV}. 
Therefore, without loss of generality, a three-dimensional (3D) Cartesian coordinate system is considered, where the BS is located at the origin and the UAV motion state vector at the $n$-th time slot can be denoted by $\mathbf{x}_{n} = [x_{n}, v_{n}^{\text{x}}, y_{n}, v_{n}^{\text{y}}]^{T}$ with $x_{n}$, $v_{n}^{\text{x}}$, $y_{n}$, and $v_{n}^{\text{y}}$ denoting the $x$-axis coordinate, the velocity along $x$-axis, the $y$-axis coordinate, and the velocity along $y$-axis, respectively. 
Despite the inherent uncertainty in UAV motion state, it is still possible to partially plan the UAV trajectory by designing the predicted state vector at the $(n+1)$-th time slot, which can be realized by remote control from the BS \cite{qqw2021JSAC,ISAC-UAV-1,ISAC-UAV-2}. 
The other parts of our considered system are specified in the following subsections. 

\begin{figure}[!t]
    \centering
    \includegraphics[width=0.38\textwidth]{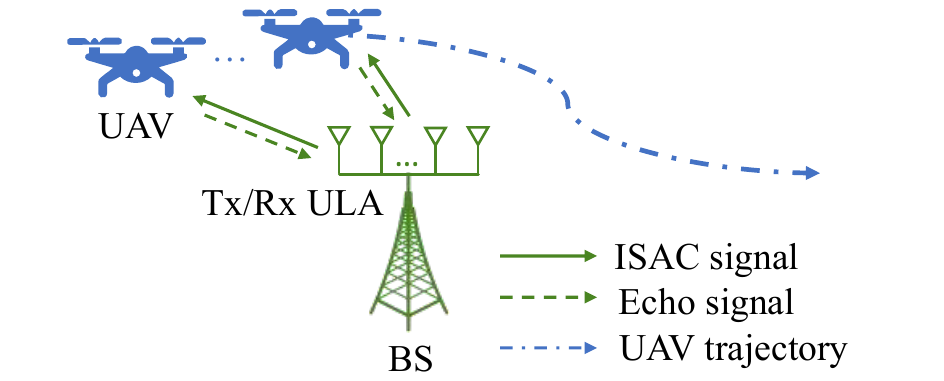}
    \caption{System model illustration.}
    \label{fig0}
    \vspace{-5mm}
\end{figure}

\subsection{UAV Mobility Model}

The entire UAV flight dynamic can be described exploiting a discrete-time state evolution model expressed as \cite{CVModel}
\vspace{-1.5mm}
\begin{equation}
    \mathbf{x}_{n} = \mathbf{G}\mathbf{x}_{n-1} + \mathbf{u}_{n} + \mathbf{z}_{\text{p},n}, \forall n \in\{1,2,...,N\}, \label{fm:SE}
    \vspace{-1.5mm}
\end{equation}
where $\mathbf{G}\in\mathbb{R}^{4\times 4}$ denotes the transition matrix, $\mathbf{u}_{n}\in\mathbb{R}^{4\times 1}$ denotes the motion control input from the BS, $N$ denotes the total number of time slots, and $\mathbf{z}_{\text{p},n} \sim \mathcal{N}(\mathbf{0},\mathbf{Q}_{\text{p}})$ denotes the process noise owing to control errors \cite{CVModel}, respectively.
The expressions of $\mathbf{G}$ and $\mathbf{Q}_{\text{p}}\in\mathbb{R}^{4\times 4}$ can be given by 
\vspace{-1.5mm}
\begin{equation}
	\mathbf{G} = \mathbf{I}_{2} \otimes \begin{bmatrix}
        1 & \Delta T  \\
        0 & 1 
    \end{bmatrix}, 
    \mathbf{Q}_{\text{p}} = \mathbf{I}_{2} \otimes \begin{bmatrix}
        \frac{1}{3}\Delta T^{3} & \frac{1}{2}\Delta T^{2}  \\
        \frac{1}{2}\Delta T^{2} & \Delta T 
    \end{bmatrix}\tilde{q},
    \vspace{-1.5mm}
\end{equation}
, respectively, where $\tilde{q}$ denotes the process noise intensity.\footnote{$\mathbf{G}$ and $\mathbf{Q}_{\text{p}}$ can be derived by sampling a continuous-time random process modelling the UAV movement as in \cite{CVModel}. Besides, the UAV movement with anisotropic accelerations can be modeled by formulating $\mathbf{Q}_{\text{p}}$ with different process noise intensities along different movement directions.} 

Note that $\{\mathbf{x}_{n}\}, \forall n,$ is indeed a random process and cannot be directly acquired by either the UAV or the BS. 
Fortunately, $\mathbf{x}_{n}$ can be predicted at the $(n-1)$-th time slot and subsequently estimated at the $n$-th time slot, which are denoted by $\breve{\mathbf{x}}_{n} = [\breve{x}_{n}, \breve{v}_{n}^{\text{x}}, \breve{y}_{n}, \breve{v}_{n}^{\text{y}}]^{T}$ and $\hat{\mathbf{x}}_{n} = [\hat{x}_{n}, \hat{v}_{n}^{\text{x}}, \hat{y}_{n}, \hat{v}_{n}^{\text{y}}]^{T}$, respectively.  
More specifically, by designing the motion control input as $\mathbf{u}_{n} = \breve{\mathbf{x}}_{n} - \mathbf{G}\hat{\mathbf{x}}_{n-1}$, the relationship between $\mathbf{x}_{n}$ and $\breve{\mathbf{x}}_{n}$ can be compactly expressed as 
\vspace{-1.5mm}
\begin{equation}
    \mathbf{x}_{n} = \breve{\mathbf{x}}_{n} + \mathbf{G}(\mathbf{x}_{n-1} - \hat{\mathbf{x}}_{n-1}) + \mathbf{z}_{\text{p},n}. \label{fm:pre}
    \vspace{-1.5mm}
\end{equation}
Therefore, the UAV trajectory optimization can be performed by appropriately optimizing $\breve{\mathbf{x}}_{n}$. 
The detailed procedures for obtaining $\hat{\mathbf{x}}_{n}$ are dependent on $\breve{\mathbf{x}}_{n}$ and are specified in the following subsection. 
Furthermore, a crucial assumption about the prediction and estimation errors is specified as follows:  

\emph{Assumption 1 (Small prediction/estimation error): In this paper, we assume that the prediction and estimation errors, although inherently exist and follow different probability distributions, are sufficiently small. 
Thus, the ground-truth value of the state vector can be approximated by the predicted and estimated values, i.e., $\mathbf{x}_{n} \approx \breve{\mathbf{x}}_{n} \approx \hat{\mathbf{x}}_{n}$ \cite{JunkunYan2015TSP,PB1,YFJ2024CL}.}

\begin{remark}
    Although assumption 1 may appear idealistic, small prediction and estimation errors are practically achievable in scenarios such as millimeter-wave ISAC systems. 
    Specifically, highly accurate localization/tracking can be achieved thanks to favorable channel conditions and large antenna array gain \cite{PB2,PB3,Relia1,Relia2}.
    Moreover, this study primarily focuses on characterizing the impact of UAV trajectories on communication reliability under the sensing-assisted predictive beamforming scheme. 
    Therefore, assumption 1 is well-justified and does not diminish the necessity and value of the proposed investigation. 
\end{remark}

{
\begin{table}[!t]
    \centering
{\caption{List of Key Notations.}} \label{tab:notation}
{
\footnotesize
\vspace{-1.7mm}
\begin{tabular}{|m{1.2cm}|m{6.6cm}|}
\hline
\textbf{Symbols} & \textbf{Description}  \\ \hline
$\mathbf{x}_{n}, \breve{\mathbf{x}}_{n}, \hat{\mathbf{x}}_{n}$ & The actual, predicted and estimated UAV motion state vector at the $n$th time slot  \\ \hline
$\mathbf{M}_{\text{p},n}, \mathbf{M}_{n}$ & The prediction and estimation MSE matrix   \\ \hline
$\xi_{\text{p},n}, \xi_{\text{e},n}$ & The random variable for the complemetary outage event at the prediction and estimation stage  \\ \hline
$\tilde{\xi}_{\text{p},n}, \tilde{\xi}_{\text{e},n}$ & The approximated random variable for the complemetary outage event   \\ \hline
$\breve{\gamma}_{n}, \hat{\gamma}_{n}$ & The target constant received SNR at the prediction and estimation stage  \\ \hline
$\bm{\gamma}_{n}$ & The target constant received SNR vector  \\ \hline
$\zeta_{\text{p},n}, \zeta_{\text{e},n}$ & The OP at the prediction and estimation stage  \\ \hline
$\tilde{\zeta}_{\text{p},n}, \tilde{\zeta}_{\text{e},n}$ & The approximated OP at the prediction and estimation stage  \\ \hline
$C_{\text{p},n}, C_{\text{e},n}$ & The outage capacity at the prediction and estimation stage  \\ \hline
$C_{n}$ & The overall outage capacity  \\ \hline
$w_{n}$ & The sensing duration ratio  \\ \hline
$\mathbf{q}_{n}, \breve{\mathbf{q}}_{n}, \hat{\mathbf{q}}_{n}$ & The ground-truth, predicted and estimated UAV trajectory  \\ \hline
$\kappa(\cdot), \varkappa(\cdot)$ & The function denoting the beam alignment accuracy and the maximum OP  \\ \hline
\end{tabular}}%
\vspace{-5mm}
\end{table}
}

\vspace{0.1mm}
\subsection{Sensing-Assisted Beamforming}

In our considered ISAC system, a two-stage predictive beamforming scheme is implemented by the BS to achieve real-time UAV tracking and communication \cite{PB1}. 
At each stage, the BS adaptively designs its beamforming vector according to the predicted or estimated UAV motion state, which is detailed as follows:

\subsubsection{Prediction Stage}
At the beginning $w_{n}$ ratio of the $n$th time slot, the BS generates the predicted state vector $\breve{\mathbf{x}}_{n}$ and transmits ISAC signals with the beamforming vector expressed as $\breve{\mathbf{f}}_{n} = \mathbf{a}(\breve{\theta}_{n}) = \mathbf{a}(\arctan(\breve{y}_{n}/\breve{x}_{n}))$, where $w_{n}$ and $\breve{\theta}_{n}$ denote the sensing duration ratio and the predicted azimuth angle, respectively. 
Based on assumption 1, predictive beamforming can achieve sufficient accuracy such that the UAV is reliably illuminated by the main lobe of the beam, enabling the BS to successfully receive echo signals from the UAV. 
Meanwhile, the BS measures the azimuth angle $\theta_{n}$ and distance $d_{n}$ of the UAV from echo signals via the matched-fitering technique \cite{MARichards}. 
The measurement model is explicitly given by 
\vspace{-1.5mm}
\begin{equation}
    \mathbf{w}_{n} 
    = \mathbf{h}(\mathbf{x}_{n}) + \mathbf{z}_{\text{m},n}   
    = \begin{bmatrix}
        \arctan(y_{n}/x_{n}) \\
        \sqrt{ x_{n}^{2} + y_{n}^{2} + H^{2} }
    \end{bmatrix}
    + \begin{bmatrix}
        z_{1,n} \\
        z_{2,n}
    \end{bmatrix}, \label{formu::SM-model}
    \vspace{-1.5mm}
\end{equation}
where $\mathbf{w}_{n} = [\hat{\theta}_{n}, \hat{d}_{n}]^{T}$ represents the measured results, $\hat{\theta}_{n}$ denotes the measured azimuth angle, $\hat{d}_{n}$ denotes the measured distance, $\mathbf{z}_{\text{m},n}$ represents the measurement noise vector with $z_{i,n} \sim \mathcal{N}(0,\sigma_{i,n}^{2}), i=1,2$, and $\sigma_{i,n}^{2}, i=1,2$ denote the corresponding measurement noise variance, respectively. 
Given the sparse blockages and scatterings in the vertical dimension, the communication channel between the BS and UAV can be assumed to be line-of-sight (LoS)-dominant with free-space path loss \cite{qqwUAV,KTMeng-2023-TWC-UAV-IPSAC,qqwMultiUAV}.\footnote{The self-interference at the ISAC BS can be suppressed by radio-frequency isolation designs and highly directional beamforming techniques \cite{Relia1}. Also, we focus on the LoS component of the echo signals because the clutter or multipath components can be mitigated by the space-time adaptive processing and LoS identification techniques \cite{MARichards,LoS-identi}.} 
Consequently, the expressions of $\sigma_{i,n}^{2}, i=1,2$ are given by 
\vspace{-1.5mm}
\begin{align}
        \sigma_{1,n}^{2} &= \frac{a_{1}^{2}( x_{n}^{2} + y_{n}^{2} + H^{2} )^{2} (x_{n}^{2}+ y_{n}^{2})}{\rho_{\text{r}} w_{n} y_{n}^{2}},  \label{fm:sgm1} \\
        \sigma_{2,n}^{2} &= \frac{a_{2}^{2}( x_{n}^{2} + y_{n}^{2} + H^{2} )^{2}}{\rho_{\text{r}} w_{n}},  \label{fm:sgm2}
    \vspace{-1.5mm}
\end{align}
and the measurement noise covariance matrix for $\mathbf{z}_{\text{m},n}$ can be derived as $\mathbf{Q}_{\text{m},n} = \mathrm{diag}(\sigma_{1,n}^{2},\sigma_{2,n}^{2})$. 
In (\ref{fm:sgm1}) and (\ref{fm:sgm2}), $a_{i},i=1,2$ represent the corresponding measurement capability coefficients calculated according to the system configurations and signal processing designs \cite{JunkunYan2015TSP,PB1}, and $\rho_{\text{r}}\in\mathbb{R}$ denotes the sensing power gain coefficient given by \cite{YFJ2024CL}
\vspace{-1.5mm}
\begin{equation}
    \rho_{\text{r}} = \frac{P_{\text{A}} N_{\text{sym}} N_{\text{t}} N_{\text{r}}}{\sigma^{2}} \left(\frac{\sigma_{\text{RCS}}\lambda^{2}}{(4\pi)^{3}}\right),
    \vspace{-1.5mm}
\end{equation}
where $P_{\text{A}}$ denotes the BS transmit power, $N_{\text{sym}}$ is matched-fitering gain accumulated during the whole time slot, $\sigma^{2}$ denotes the additive white Gaussian noise power at the receiver, $\sigma_{\text{RCS}}$ signifies the target radar cross-section, and $\lambda$ denotes the carrier wavelength \cite{MARichards}.

\subsubsection{Estimation Stage}

During the remaining period of the $n$th time slot, the BS generates the estimated state vector $\hat{\mathbf{x}}_{n}$ following the standard EKF procedures and then transmit ISAC signals with an updated transmit beamforming vector expressed as $\hat{\mathbf{f}}_{n} = \mathbf{a}(\hat{\theta}_{n}) = \mathbf{a}(\arctan(\hat{y}_{n}/\hat{x}_{n}))$ for a statistically more precise beam alignment.\footnote{Our considered predictive UAV tracking scheme can be extended to cases with other filtering techniques, such as the unscented Kalman fitering, to avoid the possibly non-negligible linearization error. 
Moreover, there might exist a time interval $T_{\text{Bu}}$ for the state estimation between the beginning of the estimation stage and the beamforming update due to practical reasons such as the computation delay.
In this case, the prediction and estimation stage duration can be calibrated as $w_{n}\Delta T + T_{\text{Bu}}$ and $(1 - w_{n})\Delta T - T_{\text{Bu}}$, respectively.
In this paper, $T_{\text{Bu}}$ and the computational overhead of EKF are assumed negligible thanks to the limited state vector dimension.}
The standard EKF procedures are given by the following steps \cite{MKay}.
\begin{enumerate}
    \item Obtaining the predicted state vector $\breve{\mathbf{x}}_{n}$.
    \item Linearization: $\mathbf{H}_{n} = \frac{\partial\mathbf{h}}{\partial\mathbf{x}_{n}}|_{\mathbf{x}_{n} = \breve{\mathbf{x}}_{n}}, \forall n$. 
    \item Calculating the prediction MSE matrix:
    \vspace{-1.5mm}
    \begin{equation}
        \mathbf{M}_{\text{p},n} = \mathbf{G}\mathbf{M}_{n-1}\mathbf{G}^{H} + \mathbf{Q}_{\text{p}}. \label{fm:Mpn}
    \vspace{-1.5mm}
    \end{equation}
    \item Calculating the Kalman gain matrix: 
    \vspace{-1.5mm}
    \begin{equation}
        \mathbf{K}_{n} = \mathbf{M}_{\text{p},n}\mathbf{H}_{n}^{H}(\mathbf{Q}_{\text{m},n} + \mathbf{H}_{n}\mathbf{M}_{\text{p},n}\mathbf{H}_{n}^{H} )^{-1}. \label{fm:efk-4}
    \vspace{-1.5mm}
    \end{equation}    
    \item Obtaining the estimated state vector: 
    \vspace{-1.5mm}
    \begin{equation}
        \hat{\mathbf{x}}_{n} = \breve{\mathbf{x}}_{n} + \mathbf{K}_{n}(\mathbf{w}_{n} - \mathbf{h}(\breve{\mathbf{x}}_{n})). \label{fm:efk-5}
    \vspace{-1.5mm}
    \end{equation} 
    \item Calculating the estimation MSE matrix: 
    \vspace{-1.5mm}
    \begin{equation}
        \!\!\!\!\!\!\!\!\!\!\!\!\mathbf{M}_{n} 
        \!=\! \left( \mathbf{I} - \mathbf{K}_{n} \mathbf{H}_{n} \right) \mathbf{M}_{\text{p},n}
        \!=\! \left( \mathbf{H}_{n}^{H} \mathbf{Q}_{\text{m},n}^{-1} \mathbf{H}_{n} + \mathbf{M}_{\text{p},n}^{-1} \right)^{-1}. \label{fm:efk-6}
    \vspace{-1.5mm}
    \end{equation}     
\end{enumerate}
The detailed derivation of (\ref{fm:efk-6}) can be referred to \cite{PB1}.
In (\ref{fm:efk-4}) and (\ref{fm:efk-6}), the expression of $\mathbf{H}_{n}$ is given by 
\vspace{-1.5mm}
\begin{align}
    \mathbf{H}_{n} = \begin{bmatrix}
        -\frac{\breve{y}_{n}}{\breve{x}_{n}^{2} + \breve{y}_{n}^{2}} & 0 & \frac{\breve{x}_{n}}{\breve{x}_{n}^{2} + \breve{y}_{n}^{2}} & 0  \\
        \frac{\breve{x}_{n}}{\sqrt{\breve{x}_{n}^{2} + \breve{y}_{n}^{2} + H^{2}}} & 0 & \frac{\breve{y}_{n}}{\sqrt{\breve{x}_{n}^{2} + \breve{y}_{n}^{2} + H^{2}}} & 0 
    \end{bmatrix}. \label{formu:Hn}
\end{align}

\subsection{Outage Capacity Characterization}

The outage capacity refers to the maximum rate maintained over the fading block such that the OP is less than a predetermined outage threshold $\varepsilon_{\rm{out}}$ \cite{goldsmith}, which can be characterized as follows in our considered system.
Given the considered predictive beamforming scheme and LoS-dominant channel model, the instantaneous received SNRs of the UAV at the prediction and estimation stage of the $n$th time slot can be represented by 
\vspace{-1.5mm}
\begin{equation}
    \!\!\!\!\!\gamma_{\text{p},n}
    = \frac{\tilde{P}|\mathbf{a}(\theta_{n})^{H}\mathbf{a}(\breve{\theta}_{n})|}{ x_{n}^{2} + y_{n}^{2} + H^{2} }, \ \ 
    \gamma_{\text{e},n}
    = \frac{\tilde{P}|\mathbf{a}(\theta_{n})^{H}\mathbf{a}(\hat{\theta}_{n})|}{ x_{n}^{2} + y_{n}^{2} + H^{2} }, \label{formu:SNR}
    \vspace{-1.5mm}
\end{equation}
respectively, where the coefficient $\tilde{P}$ is defined as $\tilde{P} \triangleq P_{\text{A}}\beta_{0}/\sigma^{2}$, and $\beta_{0} = (\lambda/4\pi)^{2}$ represents the channel power gain at the reference distance of 1 m. 
$\mathbf{a}(\cdot)$ denotes the transmitting steering vector expressed as 
\vspace{-1.5mm}
\begin{equation}
    \!\! \mathbf{a}(\theta_{n}) = [1, \mathrm{exp}(\mathrm{j}\pi\cos{\theta_{n}}), ..., \mathrm{exp}(\mathrm{j}(N_{\rm{t}} - 1)\pi\cos{\theta_{n}})]^{T}.
    \vspace{-1.5mm}
\end{equation}
Since the random factors in $\gamma_{\text{p},n}$ and $\gamma_{\text{e},n}$ (i.e., $\theta_{n}$, $\breve{\theta}_{n}$, $\hat{\theta}_{n}$, $x_{n}$, and $y_{n}$) are assumed constant within $\Delta T$ s, the BS-UAV channel fading can be modeled as slow flat fading with a coherence time of $\Delta T$ s \cite{goldsmith}.
Consequently, the complementary outage events (i.e., the events of the UAV not being in outage) at the prediction and estimation stage of the $n$th time slot are expressed as $\xi_{\text{p},n} \triangleq \gamma_{\text{p},n} - \breve{\gamma}_{n} \geq 0$ and $\xi_{\text{e},n} \triangleq \gamma_{\text{e},n} - \hat{\gamma}_{n} \geq 0$, respectively, where $\breve{\gamma}_{n}$ and $\hat{\gamma}_{n}$ denote the target constant received SNRs ensuring the OP less than $\varepsilon_{\rm{out}}$ at the corresponding stage. 
Then, the OPs at the two stages of the $n$th time slot can be uniformly expressed as
\vspace{-1.5mm}
\begin{equation}
    \!\!\!\!\zeta_{\iota,n} = \mathbb{P}\left( \xi_{\iota,n} < 0 \right) = 1 - \int_{\mathcal{Q}_{\iota,n}} f_{\iota}(\xi_{\iota,n}) \mathrm{d}\xi_{\iota,n},  \iota \in \{ \text{p}, \text{e} \}, \label{fm:OPexpress}
    \vspace{-1.5mm}
\end{equation}
where the set $\mathcal{Q}_{\iota,n} = \{\xi_{\iota,n}|\xi_{\iota,n} \geq 0\}$ is named as the complementary outage region (COR), $f_{\iota}(\cdot)$ denotes the probability density function (PDF) of $\xi_{\iota,n}$, and the subscript $\iota = \text{p}$ and $\iota = \text{e}$ denote the prediction and estimation stage, respectively.
As illustrated in Fig. \ref{fig:showCOR}, the COR is equivalent to the set constituted by all $(x_{n},y_{n})$ satisfying $\xi_{\iota,n} \geq 0$ due to (\ref{formu:SNR}).
Since $\zeta_{\iota,n}$ monotonically decreases with the decreasing of $\gamma_{\iota,n}$, $\zeta_{\text{p},n}(\breve{\gamma}_{n}) = \zeta_{\text{e},n}(\hat{\gamma}_{n}) = \varepsilon_{\rm{out}}$ hold. 
Accordingly, the outage capacities normalized by the bandwidth at the prediction and estimation stage of the $n$th time slot are expressed as 
\vspace{-1.5mm}
\begin{align}
    C_{\text{p},n} &= \log_{2}( 1 + \zeta_{\text{p},n}^{-1}(\varepsilon_{\rm{out}}) ) = \log_{2}( 1 + \breve{\gamma}_{n} ), \\
    C_{\text{e},n} &= \log_{2}( 1 + \zeta_{\text{e},n}^{-1}(\varepsilon_{\rm{out}}) ) = \log_{2}( 1 + \hat{\gamma}_{n} )
    \vspace{-1.5mm}
\end{align}
respectively, where $\zeta_{\iota,n}^{-1}(\cdot)$ denotes the inverse cumulative distribution function \cite{goldsmith}. 
The overall outage capacity at the $n$th time slot is represented by $C_{n} = w_{n}C_{\text{p},n} + (1 - w_{n})C_{\text{e},n}$.

\subsection{Problem Formulation}

In this paper, we propose a joint UAV tracking and outage capacity maximization scheme. 
To be specific, the predicted UAV trajectory $\breve{\mathbf{q}}_{n} = [\breve{x}_{n}, \breve{y}_{n}]^{T}$, sensing duration ratio $w_{n}$ and target constant received SNR vector $\bm{\gamma}_{n} = [\breve{\gamma}_{n},\hat{\gamma}_{n}]^{T}$ are jointly optimized to maximize the overall outage capacity at each time slot. 
The corresponding optimization problem is formulated as 
\vspace{-1.5mm}
\begin{align}
    (\mathrm{P1}): \ &\max_{ \{\breve{\mathbf{q}}_{n},w_{n},\bm{\gamma}_{n} \} } \ \ C_{n}  \label{opt-obj} \\
    \text{s.t.} \ 
    &\| \breve{\mathbf{q}}_{n} - \hat{\mathbf{q}}_{n-1} \| \leq v_{\text{A,max}}\Delta T, \tag{\ref{opt-obj}{a}} \label{opt-cstrt-a} \\
    &\breve{y}_{n} \geq y_{\text{min}}, \tag{\ref{opt-obj}{b}} \label{opt-cstrt-b} \\
    &w_{\text{min}} \leq w_{n} \leq w_{\text{max}}, \tag{\ref{opt-obj}{c}} \label{opt-cstrt-c} \\
    &\varkappa(\breve{\mathbf{q}}_{n},w_{n},\bm{\gamma}_{n}) \leq 0, \tag{\ref{opt-obj}{d}} \label{opt-cstrt-d} \\
    &\mathbf{0} \prec \bm{\gamma}_{n} \prec \gamma_{\text{max}}\bm{1}_{2}, \tag{\ref{opt-obj}{e}} \label{opt-cstrt-e} 
    \vspace{-1.5mm}
\end{align}
where $\hat{\mathbf{q}}_{n-1} = [\hat{x}_{n-1}, \hat{y}_{n-1}]^{T}$ denotes the estimated UAV trajectory at the $(n-1)$-th time slot, $v_{\text{A,max}}$ denotes the UAV maximum velocity, $\varkappa(\breve{\mathbf{q}}_{n},w_{n},\bm{\gamma}_{n}) = \max{ \left( \zeta_{\text{p},n}, \zeta_{\text{e},n} \right) } - \varepsilon_{\text{out}}$ represents the maximum OP at the $n$th time slot and $\gamma_{\text{max}} = \tilde{P} N_{\text{t}}/(y_{\text{min}}^{2} + H^{2})$ denotes the maximum target constant received SNR due to the maximum beamforming gain and the minimum path loss, respectively. 
In (P1), (\ref{opt-cstrt-a}) represents the maximum UAV velocity constraint, while (\ref{opt-cstrt-b}) represents a minimum $y$-axis coordinate constraint of a flyable zone.\footnote{In practice, the UAV position with $\breve{y}_{n} = 0$ leads to the infinite azimuth angle measurement noise variance. Thus, we consider a case where the UAV trajectory is constrained in an area with a nonzero minimum $y$-axis coordinate denoted by $y_{\text{min}} > 0$.} 
(\ref{opt-cstrt-c}), (\ref{opt-cstrt-d}) and (\ref{opt-cstrt-e}) denote the sensing duration ratio range, maximum tolerable OP, and SNR range constraints, respectively. 
It is non-trivial to solve (P1) because the objective function and the constraint (\ref{opt-cstrt-d}) are implicit functions of $\breve{\mathbf{q}}_{n}$, $w_{n}$, and $\bm{\gamma}_{n}$.

\begin{figure}[!t]
    \centering
    \includegraphics[width=0.4\textwidth]{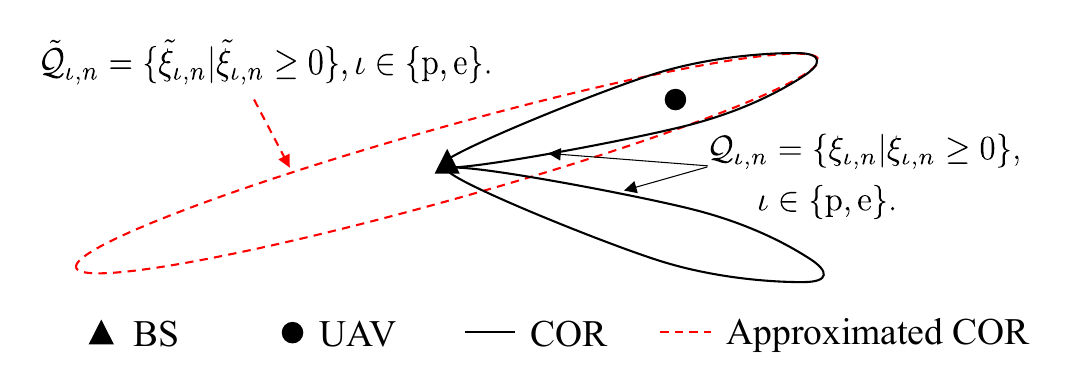}
    \caption{A geometric illustration of COR and approximated COR (aCOR).}
    \label{fig:showCOR}
    \vspace{-4mm}
\end{figure}

\section{Proposed OP Approximations}

To address the implicit objective function and constraint (\ref{opt-cstrt-d}) in (P1), OP approximations with analytical expressions are proposed in this section for the prediction and estimation stages by approximating the CORs via Taylor expansions, which makes it tractable to solve (P1).

\vspace{-0.5mm}
\subsection{Prediction Stage}

Let us denote the ground-truth UAV trajectory by $\mathbf{q}_{n} = [x_{n}, y_{n}]^{T}$. 
Then, based on assumption 1, it can be reasonably inferred that the UAV is consistently illuminated by the main lobe of the downlink transmitted beam thanks to the small prediction/estimation error. 
As a result, the beamforming gain from the BS can be expressed as
\vspace{-1.5mm}
\begin{equation}
    |\mathbf{a}(\theta_{n})^{H}\mathbf{a}(\breve{\theta}_{n})| = \frac{\sin\left(\frac{N_{\text{t}}\pi}{2} \kappa(\mathbf{q}_{n};\breve{\mathbf{q}}_{n}) \right)}{\sin\left(\frac{\pi}{2} \kappa(\mathbf{q}_{n};\breve{\mathbf{q}}_{n}) \right)}, \label{fm:bfgain}
    \vspace{-1.5mm}
\end{equation}
where the expression of $\kappa(\mathbf{q}_{n};\breve{\mathbf{q}}_{n})$ is given by
\begin{align}
    \kappa(\mathbf{q}_{n};\breve{\mathbf{q}}_{n}) 
    &\triangleq \cos{(\breve{\theta}_{n})} - \cos{(\theta_{n})} \notag \\
    &= \frac{\breve{x}_{n}}{\sqrt{\breve{x}_{n}^{2} + \breve{y}_{n}^{2}}} - \frac{x_{n}}{\sqrt{x_{n}^{2} + y_{n}^{2}}}. 
\end{align}
As such, the complementary outage event at the prediction stage, i.e., $\xi_{\text{p},n} \geq 0$, can be reformulated as
\vspace{-1.5mm}
\begin{equation}
    \frac{\sin\left(\frac{N_{\text{t}}\pi}{2} \kappa(\mathbf{q}_{n};\breve{\mathbf{q}}_{n}) \right)}{\sin\left(\frac{\pi}{2} \kappa(\mathbf{q}_{n};\breve{\mathbf{q}}_{n}) \right)} \geq \frac{\breve{\gamma}_{n}(x_{n}^{2} + y_{n}^{2} + H^{2})}{\tilde{P}}. \label{fm:cout-ineq}
    \vspace{-1.5mm}
\end{equation}
However, the left-hand side (LHS) of (\ref{fm:cout-ineq}) is intractable for calculating the integral in (\ref{fm:OPexpress}) and also overly complicated for the Taylor expansion w.r.t. $\breve{\mathbf{q}}_{n}$. 
To tackle this issue, we propose a two-step approximation detailed as follows. 

In the first step, the LHS of (\ref{fm:cout-ineq}) is approximated by its second-order Taylor expansion w.r.t. the function $\kappa(\cdot)$ at the point $\breve{\kappa}_{n} = 0$, yielding: 
\vspace{-1.5mm}
\begin{equation}
    \frac{\sin\left(\frac{N_{\text{t}}\pi}{2} \kappa(\mathbf{q}_{n};\breve{\mathbf{q}}_{n}) \right)}{\sin\left(\frac{\pi}{2} \kappa(\mathbf{q}_{n};\breve{\mathbf{q}}_{n}) \right)}
    \approx N_{\text{t}} - M \kappa(\mathbf{q}_{n};\breve{\mathbf{q}}_{n})^{2}, \label{fm:Taylor-1}
    \vspace{-1.5mm}
\end{equation}
with $M = \frac{ N_{\text{t}} \pi^{2} (N_{\text{t}}^{2} - 1) }{24}$. 
Then, the COR can be approximated by $\mathcal{Q}_{\text{p},n} \approx \mathcal{Q}_{\text{pa},n} = \{ \xi_{\text{pa},n} | \xi_{\text{pa},n} \geq 0\}$ with 
\vspace{-1.5mm}
\begin{equation}
    \xi_{\text{pa},n} = \kappa(\mathbf{q}_{n};\breve{\mathbf{q}}_{n})^{2} + \frac{\breve{\gamma}_{n}(x_{n}^{2} + y_{n}^{2} + H^{2})}{\tilde{P} M} - \frac{N_{\text{t}}}{M} \geq 0. \label{fm:cout-ineq-aprx}
    \vspace{-1.5mm}
\end{equation}
In (\ref{fm:cout-ineq-aprx}), the function $\kappa(\mathbf{q}_{n};\breve{\mathbf{q}}_{n})$ remains challenging to handle due to its fractional structure. 
Thus, the second step is to approximate $\xi_{\text{pa},n}$ by the second-order Taylor expansion w.r.t. the ground-truth UAV trajectory $\mathbf{q}_{n}$ at the point $\mathbf{q}_{n} = \breve{\mathbf{q}}_{n}$ and reformulate (\ref{fm:cout-ineq-aprx}) as
\vspace{-1.5mm}
\begin{equation}
    \xi_{\text{pa},n} \approx \tilde{\xi}_{\text{p},n} = \frac{1}{2} \acute{\mathbf{q}}_{n}^{T} \tilde{\bm{\xi}}_{\text{p},n}^{(2)} \acute{\mathbf{q}}_{n}  + (\tilde{\bm{\xi}}_{\text{p},n}^{(1)})^{T} \acute{\mathbf{q}}_{n} + \tilde{\xi}_{\text{p},n}^{(0)} \geq 0, \label{fm:cout-ineq-aprx2}
    \vspace{-1.5mm}
\end{equation}
where the vector $\acute{\mathbf{q}}_{n} = \mathbf{q}_{n} - \breve{\mathbf{q}}_{n} = [\acute{x}_{n},\acute{y}_{n}]^{T}$ represents the deviation of predicted UAV trajectory from the ground-truth UAV trajectory, $\tilde{\bm{\xi}}_{\text{p},n}^{(2)}$ and $\tilde{\bm{\xi}}_{\text{p},n}^{(1)}$ denote the Hessian matrix and gradient of the LHS of (\ref{fm:cout-ineq-aprx}) w.r.t. the ground-truth UAV trajectory $\mathbf{q}_{n}$, respectively.
The specific expressions of $\tilde{\bm{\xi}}_{\text{p},n}^{(2)}$, $\tilde{\bm{\xi}}_{\text{p},n}^{(1)}$ and $\tilde{\xi}_{\text{p},n}^{(0)}$ are given by 
\begin{align}
    \tilde{\bm{\xi}}_{\text{p},n}^{(2)} = 
    \begin{bmatrix}
        2\tilde{\xi}_{n}^{(20)} & \tilde{\xi}_{n}^{(11)} \\
        \tilde{\xi}_{n}^{(11)} & 2\tilde{\xi}_{n}^{(02)}
    \end{bmatrix}, \  \ 
    \tilde{\bm{\xi}}_{\text{p},n}^{(1)} = 
    \begin{bmatrix}
        \tilde{\xi}_{n}^{(10)} \\
        \tilde{\xi}_{n}^{(01)}
    \end{bmatrix}, 
\end{align}
with 
\begin{align}
    &\tilde{\xi}_{n}^{(20)} = \frac{\breve{y}_{n}^{4}}{ (\breve{x}_{n}^{2} + \breve{y}_{n}^{2})^{3} } + \frac{\breve{\gamma}_{n}}{M \tilde{P}}, \ 
    \tilde{\xi}_{n}^{(11)} = \frac{ -2\breve{x}_{n}\breve{y}_{n}^{3} }{ (\breve{x}_{n}^{2} + \breve{y}_{n}^{2})^{3} }, \label{fm:txi20} \\
    &\tilde{\xi}_{n}^{(02)} = \frac{\breve{x}_{n}^{2}\breve{y}_{n}^{2}}{ (\breve{x}_{n}^{2} + \breve{y}_{n}^{2})^{3} } + \frac{\breve{\gamma}_{n}}{M \tilde{P}}, \
    \tilde{\xi}_{n}^{(10)} = \frac{ 2\breve{\gamma}_{n}\breve{x}_{n} }{ M \tilde{P} }, \\
    &\tilde{\xi}_{n}^{(01)} = \frac{ 2\breve{\gamma}_{n}\breve{y}_{n} }{ M \tilde{P} }, \
    \tilde{\xi}_{\text{p},n}^{(0)} = \frac{(\breve{x}_{n}^{2} + \breve{y}_{n}^{2} + H^{2})\breve{\gamma}_{n}}{\tilde{P} M } - \frac{N_{\text{t}}}{M}. \label{fm:txi0}
\end{align}
Through (\ref{fm:Taylor-1}) and (\ref{fm:cout-ineq-aprx2}), the COR at the prediction stage can be approximated by $\mathcal{Q}_{\text{p},n} \approx \mathcal{Q}_{\text{pa},n} \approx \tilde{\mathcal{Q}}_{\text{p},n} = \{ \tilde{\xi}_{\text{p},n} | \tilde{\xi}_{\text{p},n} \geq 0 \}$, where the set $\tilde{\mathcal{Q}}_{\text{p},n}$ denotes the aCOR.
Note that it can be easily obtained that $\mathrm{det}(\tilde{\bm{\xi}}_{\text{p},n}^{(2)})>0$ holds due to $\breve{\gamma}_{n} > 0$. 
Consequently, the aCOR boundary denoted by $\tilde{\xi}_{\text{p},n} = 0$ represents an ellipse on the $(\acute{x}_{n},\acute{y}_{n})$ plane, as illustrated in Fig. \ref{fig:showCOR}.
Given the aCOR expression, our proposed approximated OP at the prediction stage of the $n$th time slot is provided in the following proposition. 

\begin{proposition}
    Given assumption 1, the OP at the prediction stage of the $n$th time slot can be approximated by 
    \begin{align}
        \zeta_{\text{p},n}
        \approx \zeta_{\text{p},n}|_{\mathcal{Q}_{\text{p},n} = \tilde{\mathcal{Q}}_{\text{p},n} }
        = \tilde{\zeta}_{\text{p},n}
        = 1 - \mathbb{E}_{\acute{x}_{n}}\left[ \breve{\chi}(\acute{x}_{n}) \right], \label{fm:aCOP}
    \end{align}
    with 
    \begin{align}
        \!\!\!\!\breve{\chi}(\acute{x}_{n}) \triangleq 
        \left\{ \begin{aligned}
            &\frac{\mathrm{erf}\left(  \breve{\chi}_{\text{U}}(\acute{x}_{n}) \right) - \mathrm{erf}\left(  \breve{\chi}_{\text{L}}(\acute{x}_{n}) \right)}{2}, \!\!\!\!&\acute{x}_{n} \in [\acute{x}_{\text{L}}, \acute{x}_{\text{U}}], \\
            &0, \!\!\!\!&\text{Otherwise}, 
        \end{aligned}
        \right. \label{fm:abOP-1}
    \end{align}
    where $\acute{x}_{n}$ follows a zero mean Gaussian distribution with a variance of $\breve{\Lambda}_{\text{x},n}^{2} = [\mathbf{M}_{\text{p},n}]_{11}$. 
    The specific expressions of $\breve{\chi}_{\text{U}}(\cdot)$, $\breve{\chi}_{\text{L}}(\cdot)$, $\acute{x}_{\text{U}}$, and $\acute{x}_{\text{L}}$ are given by (\ref{fm:abOP-2})-(\ref{fm:abOP-4}), respectively.
\end{proposition}
\begin{proof}
    Please refer to Appendix A. 
\end{proof}

\begin{figure*}[!t]
    \vspace{-5mm}
    \begin{align}
        &\breve{\chi}_{\text{U}}(\acute{x}_{n}) 
        = \frac{ \breve{\Lambda}_{\text{x},n}^{2}\acute{y}_{\text{U}}(\acute{x}_{n}) - \breve{\Lambda}_{\text{xy},n}^{2} \acute{x}_{n} }{\sqrt{2|\det{(\bm{\breve{\Lambda}}_{n})}|}\breve{\Lambda}_{\text{x},n}}, 
        \breve{\chi}_{\text{L}}(\acute{x}_{n}) 
        = \frac{ \breve{\Lambda}_{\text{x},n}^{2}\acute{y}_{\text{L}}(\acute{x}_{n}) - \breve{\Lambda}_{\text{xy},n}^{2} \acute{x}_{n} }{\sqrt{2|\det{(\bm{\breve{\Lambda}}_{n})}|}\breve{\Lambda}_{\text{x},n}}, 
        \acute{x}_{\text{U}} = -\breve{x}_{n} + \sqrt{ -\frac{Y_{1}}{Y_{2}} },
        \acute{x}_{\text{L}} = -\breve{x}_{n} - \sqrt{ -\frac{Y_{1}}{Y_{2}} },  \label{fm:abOP-2} \\
        &\acute{y}_{\text{U}}(\acute{x}_{n}) 
        = -\breve{y}_{n} + Y_{0}(\acute{x}_{n} + \breve{x}_{n}) + \sqrt{Y_{1} + Y_{2}(\acute{x}_{n} + \breve{x}_{n})^{2}}, 
        \acute{y}_{\text{L}}(\acute{x}_{n}) 
        = -\breve{y}_{n} + Y_{0}(\acute{x}_{n} + \breve{x}_{n}) - \sqrt{Y_{1} + Y_{2}(\acute{x}_{n} + \breve{x}_{n})^{2}}, \\
        &Y_{0} = \frac{\breve{x}_{n}\breve{y}_{n}^{3}\tilde{P} M}{(\breve{x}_{n}^{2} + \breve{y}_{n}^{2})^{3}\breve{\gamma}_{n} + \breve{x}_{n}^{2} \breve{y}_{n}^{2} \tilde{P} M}, 
        Y_{1} = \frac{(\tilde{P} N_{\text{t}} - H^{2} \breve{\gamma}_{n})(\breve{x}_{n}^{2} + \breve{y}_{n}^{2})^{3}}{(\breve{x}_{n}^{2} + \breve{y}_{n}^{2})^{3}\breve{\gamma}_{n} + \breve{x}_{n}^{2} \breve{y}_{n}^{2} \tilde{P} M}, 
        Y_{2} = -\frac{(\breve{x}_{n}^{2} + \breve{y}_{n}^{2})^{4}((\breve{x}_{n}^{2} + \breve{y}_{n}^{2})^{2}\breve{\gamma}_{n}^{2} + \tilde{P} M \breve{y}_{n}^{2}\breve{\gamma}_{n} )}{((\breve{x}_{n}^{2} + \breve{y}_{n}^{2})^{3}\breve{\gamma}_{n} + \breve{x}_{n}^{2} \breve{y}_{n}^{2} \tilde{P} M)^{2}}, \label{fm:abOP-4} 
        \vspace{-5mm}
    \end{align}
    \rule{18.2cm}{0.5pt}
    \vspace{-10mm}
\end{figure*}

\vspace{-3mm}
\subsection{Estimation Stage}
Let us denote the estimated UAV trajectory at the $n$th time slot by $\hat{\mathbf{q}}_{n} = [\hat{x}_{n}, \hat{y}_{n}]^{T}$ and define $\grave{\mathbf{q}}_{n}$ as $\grave{\mathbf{q}}_{n} \triangleq \mathbf{q}_{n} - \hat{\mathbf{q}}_{n} = [\grave{x}_{n}, \grave{y}_{n}]^{T}$.
Note that $\hat{\mathbf{q}}_{n}$ is unknown at the $(n-1)$th time slot and only available after receiving the echo signals at the $n$th time slot. 
To address this issue, the approximation $\hat{\mathbf{q}}_{n} \approx \breve{\mathbf{q}}_{n}$ is reasonably applied thanks to assumption 1 so that the OP at the estimation stage of the $n$th time slot can be approximately calculated at the $(n-1)$th time slot. 
Then, similar as the derivation process from (\ref{fm:bfgain}) to (\ref{fm:cout-ineq-aprx2}), the aCOR at the estimation stage of the $n$th time slot can be expressed as $\tilde{\mathcal{Q}}_{\text{e},n} = \{ \tilde{\xi}_{\text{e},n} | \tilde{\xi}_{\text{e},n} \geq 0 \}$ with 
\begin{align}
    \tilde{\xi}_{\text{e},n} = \frac{1}{2} \grave{\mathbf{q}}_{n}^{T} \tilde{\bm{\xi}}_{\text{e},n}^{(2)} \grave{\mathbf{q}}_{n}  + (\tilde{\bm{\xi}}_{\text{e},n}^{(1)})^{T} \grave{\mathbf{q}}_{n} + \tilde{\xi}_{\text{e},n}^{(0)},
\end{align}
where $\tilde{\bm{\xi}}_{\text{e},n}^{(2)}$, $\tilde{\bm{\xi}}_{\text{e},n}^{(1)}$, and $\tilde{\xi}_{\text{e},n}^{(0)}$ are given by $\tilde{\bm{\xi}}_{\text{e},n}^{(2)} = \tilde{\bm{\xi}}_{\text{p},n}^{(2)}|_{\breve{\gamma}_{n} = \hat{\gamma}_{n}}$, $\tilde{\bm{\xi}}_{\text{e},n}^{(1)} = \tilde{\bm{\xi}}_{\text{p},n}^{(1)}|_{\breve{\gamma}_{n} = \hat{\gamma}_{n}}$, and $\tilde{\xi}_{\text{e},n}^{(0)} = \tilde{\xi}_{\text{p},n}^{(0)}|_{\breve{\gamma}_{n} = \hat{\gamma}_{n}}$, respectively. 

Meanwhile, similar as the derivation process in the proof of Proposition 1, $\grave{\mathbf{q}}_{n} \sim \mathcal{N}(\mathbf{0},\bm{\hat{\Lambda}}_{n})$ approximately holds with 
\begin{align}
    \bm{\hat{\Lambda}}_{n} &= 
    \begin{bmatrix}
        \hat{\Lambda}_{\text{x},n}^{2} & \hat{\Lambda}_{\text{xy},n}^{2} \\
        \hat{\Lambda}_{\text{xy},n}^{2} & \hat{\Lambda}_{\text{y},n}^{2}
    \end{bmatrix} 
    \approx 
    \begin{bmatrix}
        [\mathbf{M}_{n}]_{11} & [\mathbf{M}_{n}]_{13} \\
        [\mathbf{M}_{n}]_{31} & [\mathbf{M}_{n}]_{33}
    \end{bmatrix},
\end{align}
and the OP at the estimation stage of the $n$th time slot can be approximated by 
\vspace{-1.5mm}
\begin{equation}
    \zeta_{\text{e},n} 
    \approx \zeta_{\text{p},n}|_{\mathcal{Q}_{\text{e},n} = \tilde{\mathcal{Q}}_{\text{e},n}}
    = \tilde{\zeta}_{\text{e},n} = 1 - \mathbb{E}_{\grave{x}_{n}}\left[ \hat{\chi}(\grave{x}_{n}) \right], \label{fm:ahCOP}
    \vspace{-1.5mm}
\end{equation}
with 
\begin{align}
    \!\!\!\!\!\!\hat{\chi}(\grave{x}_{n}) \triangleq 
        \left\{ \begin{aligned}
            &\frac{\mathrm{erf}\left(  \hat{\chi}_{\text{U}}(\grave{x}_{n}) \right) - \mathrm{erf}\left(  \hat{\chi}_{\text{L}}(\grave{x}_{n}) \right)}{2}, \!\!\!\!\!\!\!\!\!\!\!\!\!\!\!\!\!\!\!\!\!\!\!\!\!\!\!\!\!\!\!\!\!\!\!\! &\grave{x}_{n} \in [\grave{x}_{\text{L}}, \grave{x}_{\text{U}}], \\
            &0, \!\!\!\!\!\!\!\!\!\!\!\!\!\!\!\!\!\!\!\!\!\!\!\!\!\!\!\!\!\!\!\!\!\!\!\! &\grave{x}_{n} \in (-\infty, \grave{x}_{\text{L}}) \cup (\grave{x}_{\text{U}}, \infty). 
        \end{aligned}
        \right. \label{fm:ahOP-1} 
\end{align}
The expressions of functions $\hat{\chi}_{\text{U}}(\grave{x}_{n})$ and $\hat{\chi}_{\text{L}}(\grave{x}_{n})$ are given by 
\begin{align}
    \hat{\chi}_{\text{U}}(\grave{x}_{n}) 
    &= \frac{ \hat{\Lambda}_{\text{x},n}^{2}\grave{y}_{\text{U}}(\grave{x}_{n}) - \hat{\Lambda}_{\text{xy},n}^{2} \grave{x}_{n} }{\sqrt{2|\det{(\bm{\hat{\Lambda}}_{n})}|}\hat{\Lambda}_{\text{x},n}}, \\
    \hat{\chi}_{\text{L}}(\grave{x}_{n}) 
    &= \frac{ \hat{\Lambda}_{\text{x},n}^{2}\grave{y}_{\text{L}}(\grave{x}_{n}) - \hat{\Lambda}_{\text{xy},n}^{2} \grave{x}_{n} }{\sqrt{2|\det{(\bm{\hat{\Lambda}}_{n})}|}\hat{\Lambda}_{\text{x},n}},
\end{align}
with $\grave{x}_{\text{L}} = \acute{x}_{\text{L}}|_{\breve{\gamma}_{n} = \hat{\gamma}_{n}}$, $\grave{x}_{\text{U}} = \acute{x}_{\text{U}}|_{\breve{\gamma}_{n} = \hat{\gamma}_{n}}$, $\grave{y}_{\text{L}} = \acute{y}_{\text{L}}|_{\breve{\gamma}_{n} = \hat{\gamma}_{n}}$ and $\grave{y}_{\text{U}} = \acute{y}_{\text{U}}|_{\breve{\gamma}_{n} = \hat{\gamma}_{n}}$, respectively.

\begin{remark}
    Estentially, our proposed OP approximations are derived from aCORs and thus the approximation accuracies mainly depend on the two-step Taylor expansion approximation accuracies. 
    Although the aCOR seems quite different from the COR as shown in Fig. \ref{fig:showCOR}, our proposed approximations still achieve satisfactory accuracies when the PDFs of $\breve{\mathbf{q}}_{n}$ and $\hat{\mathbf{q}}_{n}$ are highly concentrated in a neighborhood of $\mathbf{q}_{n}$ contained by both the COR and the aCOR, in which $\|\breve{\mathbf{q}}_{n} - \mathbf{q}_{n}\|$ and $\|\hat{\mathbf{q}}_{n} - \mathbf{q}_{n}\|$ are sufficiently small. 
    Fortunately, the existence of such neighborhood is theoretically guaranteed by assumption 1 and the property of Taylor expansions, which validates our proposed approximations. 
    Simulations further verify the approximation accuracies in Section V. 
    Furthermore, (\ref{fm:aCOP}) and (\ref{fm:ahCOP}) show that \emph{the sensing accuracy characterized by MSE matrices decides the integral over aCOR, which further influences the communication reliability}.
\end{remark}

\section{Proposed Algorithms}

Given the approximated OPs presented in (\ref{fm:aCOP}) and (\ref{fm:ahCOP}), (P1) can be reformulated into an approximated optimization problem as:  
\vspace{-1.5mm}
\begin{align}
    (\mathrm{P2}): \ &\max_{ \{ \breve{\mathbf{q}}_{n}, w_{n}, \bm{\gamma}_{n} \} } \ \ C_{n} \label{opt2-obj} \\
    \text{s.t.} \ 
    &\text{(\ref{opt-cstrt-a})-(\ref{opt-cstrt-c}), (\ref{opt-cstrt-e})}, \notag \\
    &\tilde{\varkappa}(\breve{\mathbf{q}}_{n}, w_{n}, \bm{\gamma}_{n}) \leq 0, \tag{\ref{opt2-obj}{a}} \label{opt2-cstrt-a} 
    \vspace{-1.5mm}
\end{align}
with $\tilde{\varkappa}(\breve{\mathbf{q}}_{n}, w_{n},\bm{\gamma}_{n}) = \max{ ( \tilde{\zeta}_{\text{p},n}, \tilde{\zeta}_{\text{e},n} ) } - \varepsilon_{\text{out}}$. 
Compared with (P1), the original implicit constraint (\ref{opt-cstrt-d}) in (P1) has been replaced by the approximated outage constraint (\ref{opt2-cstrt-a}). 
However, (P2) remains challenging to be optimally solved due to non-convex constraint (\ref{opt2-cstrt-a}) and the coupling among the optimization variables. 
To address this issue, an algorithm based on bisection search is proposed to obtain an efficient solution to (P2) with guaranteed convergence. 
To further reduce computational complexity, a second efficient algorithm is proposed based on AO. 

\begin{algorithm}[!t]
	\caption{Proposed search-based algorithm for (P2).}
	{\begin{algorithmic}[1] 
		\STATE {Initialize the case indicator $l_{\text{C}}$ ($l_{\text{C}} = 0$ for the case with $C_{\text{p},n} \geq C_{\text{e},n}$ or $l_{\text{C}} = 1$ otherwise), the error tolerance $\delta_{\text{C}}$, $C_{\text{max}} = \log_{2}(1+\gamma_{\rm{max}})$, and $i=1$.} 
        \STATE {Set the searching range $[C_{i}^{\text{L}}, C_{i}^{\text{U}}] = [0, C_{\text{max}}]$.}
        \FOR {$l_{\text{C}} = 0$ to $1$}
            \WHILE{$| C_{i}^{\text{U}} - C_{i}^{\text{L}} | \leq \delta_{\text{C}}$ is satisfied.}
            \STATE {Obtain $C_{i} = (C_{i}^{\text{L}} + C_{i}^{\text{U}})/2$ and solve (P2.1) and (P2.2) iteratively by Algorithm 2 given $C_{i}$ and $l_{\text{C}}$.}
            \STATE {Obtain the outputted feasibility indicator $l_{\text{f}}$.}
            \IF {$l_{\text{f}} = 0$}
                \STATE {Update $[C_{i+1}^{\text{L}}, C_{i+1}^{\text{U}}] = [C_{i}^{\text{L}}, C_{i}]$.}
            \ELSE 
                \STATE {Update $[C_{i+1}^{\text{L}}, C_{i+1}^{\text{U}}] = [C_{i}, C_{i}^{\text{U}}]$.}
            \ENDIF
            \STATE {Update $i = i+1$.}
        \ENDWHILE
        \ENDFOR
        \STATE {Compare $C_{i}$ between the case with $l_{\text{C}} = 0$ and $l_{\text{C}} = 1$ and output the larger value as the maximized outage capacity.}
	\end{algorithmic}}
\end{algorithm}

\subsection{Search-Based Algorithm} 

To decouple $\breve{\mathbf{q}}_{n}$ from $w_{n}$ and $\bm{\gamma}_{n}$, our proposed search-based algorithm solves (P2) by iteratively solving two subproblems formulated as 
\vspace{-1.5mm}
\begin{align}
    (\mathrm{P2.1}): \ &\mathrm{Find} \ \ w_{n}, \bm{\gamma}_{n}  \label{opt21-obj}  \\
    \text{s.t.} \ 
    &\text{(\ref{opt-cstrt-c}), (\ref{opt-cstrt-e})}, \notag \\
    &C_{n} = C_{i}, \tag{\ref{opt21-obj}{a}} \label{opt21-cstrt-a} 
    \vspace{-1.5mm}
\end{align}
and
\vspace{-1.5mm}
\begin{align}
    (\mathrm{P2.2}): \ \min_{ \breve{\mathbf{q}}_{n} } \ \ \tilde{\varkappa}(\breve{\mathbf{q}}_{n}, w_{i}, \bm{\gamma}_{i}) \ \ 
    \text{s.t.} \ \  \text{(\ref{opt-cstrt-a}), (\ref{opt-cstrt-b}),} \notag 
    \vspace{-1.5mm}
\end{align}
respectively, where $C_{i}$ and $(w_{i},\bm{\gamma}_{i})$ denote a given objective value and the solution to (P2.1) in the $i$-th iteration, respectively. 
Our proposed search-based algorithm solves (P2.1) to generate a candidate solution $(w_{i},\bm{\gamma}_{i})$ and subsequently evaluate its feasibility to (P2) by solving (P2.2) in the $i$-th iteration, as summarized in Algorithm 1. 

To address the non-convex constraint (\ref{opt21-cstrt-a}), (P2.1) can be further divided into two subproblems with $C_{\text{p},n} \geq C_{\text{e},n}$ and $C_{\text{p},n} < C_{\text{e},n}$, respectively. 
In both cases, $C_{n}$ is the monotonical function of $\bm{\gamma}_{n}$ and $w_{n}$. 
Then, the updating rule of the given objective value is designed based on the monotonicity of $\tilde{\varkappa}(\breve{\mathbf{q}}_{n}, w_{n}, \bm{\gamma}_{n})$ w.r.t. $\bm{\gamma}_{n}$ given in the following proposition.\footnote{In algorithm 1-3, the subscript $n$ is omitted for notational simplicity without ambiguity.}

\begin{proposition}
    Given any feasible $\breve{\mathbf{q}}_{i},w_{i}$, $\tilde{\varkappa}(\breve{\mathbf{q}}_{i}, w_{i}, \bm{\gamma}_{n})$ is a monotonically nondecreasing function of $\bm{\gamma}_{n}$.
\end{proposition}
\begin{proof}
    Please refer to Appendix B. 
\end{proof}

Proposition 2 indicates that it is easier to identify feasible solutions to (P2) with smaller target constant received SNRs. 
Let us denote the feasible set for (P2) given $\bm{\gamma}$ by 
\vspace{-1.5mm}
\begin{equation}
    \mathcal{S}_{i}(\bm{\gamma}) = \{ (w_{n}, \breve{\mathbf{q}}_{n}) | \tilde{\varkappa}(\breve{\mathbf{q}}_{n}, w_{n}, \bm{\gamma}) < 0 \}, 
    \vspace{-1.5mm}
\end{equation}
respectively. 
Then, for different $\bm{\gamma}$ and $\bm{\gamma}^{'}$, $\mathcal{S}_{i}(\bm{\gamma}) \subseteq \mathcal{S}_{i}(\bm{\gamma}^{'})$ holds if $\bm{\gamma}_{i} \preceq \bm{\gamma}_{i}^{'}$ is satisfied due to 
\vspace{-1.5mm}
\begin{equation}
    \tilde{\varkappa}(\breve{\mathbf{q}}_{n}, w_{n}, \bm{\gamma}_{i}^{'}) \leq \tilde{\varkappa}(\breve{\mathbf{q}}_{n}, w_{n}, \bm{\gamma}_{i}) \leq 0, \forall \bm{\gamma}_{i} \in \mathcal{S}_{i}(\bm{\gamma}). 
    \vspace{-1.5mm}
\end{equation}
Thanks to the continuity of $\tilde{\varkappa}(\breve{\mathbf{q}}_{n}, w_{n}, \bm{\gamma}_{n})$ w.r.t $\breve{\mathbf{q}}_{n}$ in most parts of COR, the feasible set $\mathcal{S}_{i}(\bm{\gamma}^{'})$ probably contains elements not belonging to $\mathcal{S}_{i}(\bm{\gamma})$, which further indicates a higher chance of finding a feasible solution to (P2) in $\mathcal{S}_{i}(\bm{\gamma}^{'})$ than in $\mathcal{S}_{i}(\bm{\gamma})$. 
Inspired by this observation, a two-layer bisection search is applied to solve (P2.1) ensuring the convergence. 

To address the non-convex objective function in (P2.2), the subalgorithm for (P2.2) applies the SCA technique to efficiently obtain a locally optimal solution \cite{KTMeng-2023-TWC-UAV-IPSAC}. 
In the $m$-th iteration, (P2.2) is solved with the objective function replaced by a surrogate function based on the second-order Taylor expansion, given by 
\begin{align}
    &\tilde{\varkappa}_{\text{a}}(\breve{\mathbf{q}}_{m};\breve{\mathbf{q}}_{m-1}^{\text{E}},w_{j},\bm{\gamma}_{k}) 
    = \tilde{\varkappa}(\breve{\mathbf{q}}_{m-1}^{\text{E}};w_{j},\bm{\gamma}_{k}) \notag \\
    &+ \nabla \tilde{\varkappa}(\breve{\mathbf{q}}_{m-1}^{\text{E}};w_{j},\bm{\gamma}_{k})^{T} (\breve{\mathbf{q}}_{m} - \breve{\mathbf{q}}_{m-1}^{\text{E}}) \notag \\
    &+ Q\| \breve{\mathbf{q}}_{m} - \breve{\mathbf{q}}_{m-1}^{\text{E}} \|^{2}, \label{fm:sca}
\end{align}
where $\breve{\mathbf{q}}_{m-1}^{\text{E}}$ denotes the Taylor expansion point in the $(m-1)$-th iteration and $Q$ is a given positive real number ensuring the convexity of (\ref{fm:sca}). 
Problem (P2.2) with the objective function replaced by (\ref{fm:sca}) is a convex optimization problem and can be optimally solved by standard numerical convex programming solvers such as CVX tools \cite{cvxr}. 
The overall algorithm iteratively solving (P2.1) and (P2.2) is summarized in Algorithm 2.

\vspace{-3mm}
\subsection{AO-Based Algorithm}
The computational overhead of our proposed search-based algorithm mainly exists in the iteration number $k$ owing to the complicated function $\tilde{\varkappa}(\cdot)$ and trials of infeasible $(w_{j},\bm{\gamma}_{k})$.
To significantly reduce these redundant computations, a second algorithm for (P2) is proposed based on the AO method, where the obtained $w_{n}$ and $\bm{\gamma}_{n}$ are always feasible to (P2) in each iteration. 
To be specific, given a feasible predicted UAV trajectory $\breve{\mathbf{q}}_{n} = \breve{\mathbf{q}}_{i}$ in the $i$-th iteration of our AO-based algorithm, (P2) is simplified as a subproblem formulated as
\vspace{-1.5mm}
\begin{align}
    (\mathrm{P3.1}): \ &\max_{ \{ w_{n}, \bm{\gamma}_{n} \} } \ \ C_{n}  \label{opt31-obj} \\
    \text{s.t.} \ 
    &\text{(\ref{opt-cstrt-c}), (\ref{opt-cstrt-e})}, \notag \\
    &\tilde{\varkappa}(\breve{\mathbf{q}}_{i}, w_{n}, \bm{\gamma}_{n}) \leq 0. \tag{\ref{opt31-obj}{a}} \label{opt31-cstrt-a}
    \vspace{-1.5mm}
\end{align}
To handle the non-convex objective function and constraint (\ref{opt31-cstrt-a}) in (P3.1), $w_{n}$ can be heuristically searched by a one-dimensional search, such as the golden section search \cite{goldsec}. 
As such, the subproblem of (P3.1) given the searched $w_{n}$ is a convex optimization problem thanks to Proposition 2, and thus can be optimally solved by the bisection search \cite{CVX}. 
Then, as summarized in Algorithm 3, given the obtained solution to (P3.1) denoted by $(w_{i}^{*},\bm{\gamma}_{i}^{*})$ in the $i$-th iteration, our proposed algorithm solves (P2.2) and updates $\breve{\mathbf{q}}_{i+1}$ by the obtained solution to (P2.2).

\begin{algorithm}[!t]
	\caption{Overall algorithm for (P2.1) and (P2.2).}
    {
    \begin{algorithmic}[1] 
		\STATE {Initialize $C_{i}$, $l_{\text{C}}$,  $j=1$, the error tolerance $\delta_{\text{w}}$, $\delta_{\text{p}}$, the searching range $[w_{j}^{\text{L}},w_{j}^{\text{U}}] = [w_{\text{min}}, w_{\text{max}}]$, and feasibility indicator $l_{\text{f}} = 0$.} 
        \WHILE{$| w_{j}^{\text{U}} - w_{j}^{\text{L}} | > \delta_{\text{w}}$ and $l_{\text{f}} = 0$}
            \STATE {Obtain $w_{j} = (w_{j}^{\text{L}} + w_{j}^{\text{U}})/2$ and set $k=1$.} 
            \STATE {Set the searching range $[C_{\text{p},k}^{\text{L}},C_{\text{p},k}^{\text{U}}] = [C_{i}, C_{\text{max}}]$ for the case with $l_{\text{C}} = 0$, or $[C_{\text{p},k}^{\text{L}},C_{\text{p},k}^{\text{U}}] = [0, C_{i}]$ otherwise.}
            \WHILE{$| C_{\text{p},k}^{\text{U}} - C_{\text{p},k}^{\text{L}} | > \delta_{\text{p}}$ and $l_{\text{f}} = 0$}
                \STATE {Obtain $C_{\text{p},k} = (C_{\text{p},k}^{\text{L}} + C_{\text{p},k}^{\text{U}})/2$ and $C_{\text{e},k}$ from (\ref{opt21-cstrt-a}).} 
                \STATE {Obtain $\bm{\gamma}_{k} = [\breve{\gamma}_{k}, \hat{\gamma}_{k}]^{T} = [2^{C_{\text{p},k}} - 1, 2^{C_{\text{e},k}} - 1]^{T}$.}
                \STATE {Solve (P2.2) by SCA with $(w_{j}, \bm{\gamma}_{k})$ to obtain $\breve{\mathbf{q}}_{k}$.}
                \IF{$\tilde{\varkappa}(\breve{\mathbf{q}}_{k}, w_{j}, \bm{\gamma}_{k}) < 0$}
                    \STATE {Set $l_{\text{f}} = 1$ and $(\breve{\mathbf{q}}_{i}^{*}, w_{i}^{*}, \bm{\gamma}_{i}^{*}) = (\breve{\mathbf{q}}_{k}, w_{j}, \bm{\gamma}_{k})$.}
                \ELSE 
                    \STATE {Update $[C_{\text{p},k+1}^{\text{L}}, C_{\text{p},k+1}^{\text{U}}] = [C_{\text{p},k}^{\text{L}}, C_{\text{p},k}]$ for the case with $\tilde{\zeta}_{\text{p},n}(\breve{\mathbf{q}}_{k}, w_{j}, \breve{\gamma}_{k}) \geq 0$, or $[C_{\text{p},k+1}^{\text{L}}, C_{\text{p},k+1}^{\text{U}}] = [C_{\text{p},k}, C_{\text{p},k}^{\text{U}}]$ otherwise.}
                \ENDIF
                \STATE {Update $k = k + 1$.}
            \ENDWHILE
            \IF{$l_{\text{f}} = 0$}
                \STATE {Update $[w_{j+1}^{\text{L}}, w_{j+1}^{\text{U}}] = [w_{j}, w_{j+1}^{\text{U}}]$ for the case with $l_{\text{C}} = 0$, or $[w_{j+1}^{\text{L}}, w_{j+1}^{\text{U}}] = [w_{j}^{\text{L}}, w_{j}]$ otherwise.}
            \ENDIF
            \STATE {Update $j = j + 1$.}
        \ENDWHILE
        \STATE {Output $l_{\text{f}}$ and additionally $(\breve{\mathbf{q}}_{i}^{*}, w_{i}^{*}, \bm{\gamma}_{i}^{*})$ if $l_{\text{f}} = 1$ holds.}
	\end{algorithmic}}
\end{algorithm}
\vspace{-3mm}

\subsection{Convergence and Computational Complexity Analysis}

The computational complexities of our proposed search-based and AO-based algorithm can be analyzed as follows.  
Specifically, the number of iterations needed for the convergence of the bisection search for $w_{n}$ and $C_{\text{p},n}$ can be given by $I_{\text{w}} = \log_{2}\left( \lfloor (w_{\text{max}} - w_{\text{min}})/\epsilon_{\text{w}} \rfloor \right)$ and $I_{\text{C}} = \log_{2}\left( \lfloor \log_{2}{(1+\gamma_{\text{max}})}/\epsilon_{\text{C}} \rfloor \right)$, respectively, where $\epsilon_{\text{w}}$ and $\epsilon_{\text{C}}$ denotes the tolerance of the bisection search for $w_{n}$ and $C_{\text{p},n}$, respectively \cite{CVX}. 
Thus, the computational complexity of our proposed search-based algorithm can be given by $\mathcal{O}(2I_{\text{w}}I_{\text{C}}^{2}J_{\text{A}})$, where $J_{\text{A}}$ represents the number of iterations needed for the convergence of the SCA to solve (P2.2). 
In comparison, the computational complexity of our proposed AO-based algorithm can be given by $\mathcal{O}(I_{\text{w}}^{'}I_{\text{C}} + J_{\text{A}})$, where $I_{\text{w}}^{'}$ denotes the number of iterations needed for the convergence of the one-dimensional search for $w_{n}$. 
Assuming $I_{\text{w}} \approx I_{\text{w}}^{'}$, the computational complexity of our proposed AO-based algorithm is generally lower than that of our search-based algorithm. 
Nevertheless, the convergence of the search-based algorithm is guaranteed thanks to the guaranteed convergence of the bisection search while the convergence of the AO-based algorithm is not guaranteed owing to the heuristic search for $\breve{\mathbf{q}}_{n}$. 
To ensure practical applicability, a maximum number of iterations can be predetermined to force the termination of AO-based algorithm. 

\begin{algorithm}[!t]
    \caption{Proposed AO-based algorithm for (P2).}
	{\begin{algorithmic}[1] 
        \STATE {Initialize the maximum iteration number $I_{\text{max}}$ and a solution $\breve{\mathbf{q}}_{0}^{*}$ feasible to (P2). Set the iteration number $i=1$.} 
        \WHILE {$i \leq I_{\text{max}}$}
            \STATE {Solve (P3.1) given $\breve{\mathbf{q}}_{i-1}^{*}$ by searching $w_{n}$ via the golden section in the outer layer and searching $\bm{\gamma}_{i}$ via the bisection search in the inner layer. Obtain the solution $(w_{i}^{*},\bm{\gamma}_{i}^{*})$.}
            \STATE {Solve (P2.2) given $(w_{i}^{*},\bm{\gamma}_{i}^{*})$ by SCA and obtain the solution $\breve{\mathbf{q}}_{i}^{*}$.}
            \STATE {Update $i = i + 1$.}
        \ENDWHILE
	\end{algorithmic}}
\end{algorithm}

\section{Simulation Results}

In this section, numerical results are provided to verify the effectiveness of proposed OP approximations and algorithms. 
Unless specified otherwise, the following system parameters are used: $P_{\text{A}} = 0.1 \text{W}$, $\sigma_{\text{RCS}} = 0.2 \text{m}^{2}$, $\lambda = 0.01 \text{m}$, $\sigma^{2} = -80 \text{dBm}$, $H = 50 \text{m}$, $\Delta T = 0.02 \text{s}$, $v_{\text{A,max}} = 30 \text{m/s}$, $\tilde{q} = 10^{-3}$, $N_{\text{sym}} = 10^{4}$, $N = 10^{3}$, $w_{\text{min}} = 0.1$, and $w_{\text{max}} = 1$ \cite{KTM-IOS,Relia2,YFJ2024CL}.

\begin{figure*}[!t]
	\centering
	\vspace{-4mm}
    \subfigure[OPs at the prediction stage.]{
        \vspace{-1mm}
        \label{smgA:a}
		\includegraphics[width=0.3\textwidth]{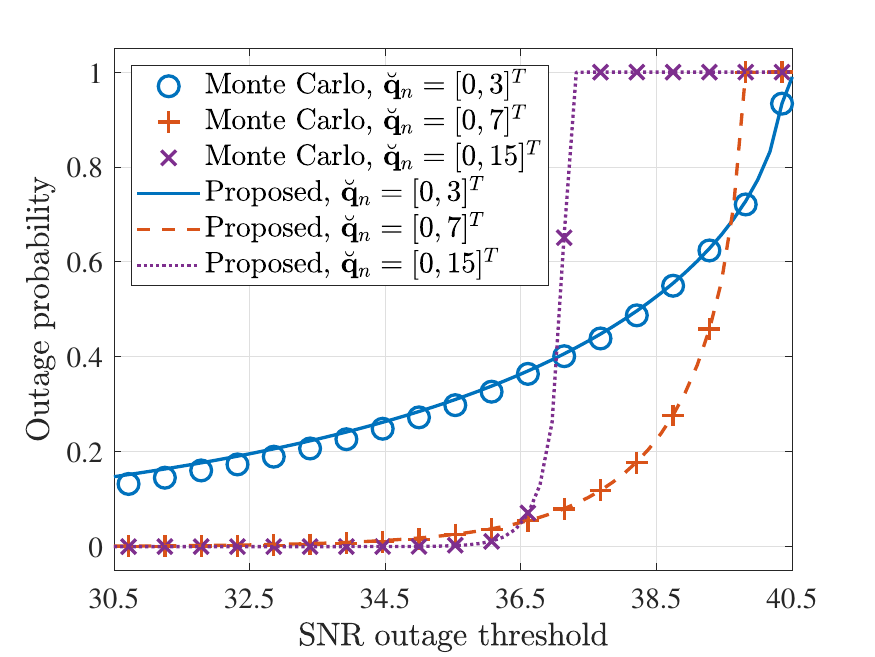}
	} 
    \subfigure[OPs at the estimation stage.]{
		\vspace{-1mm}
        \label{smgA:b}
		\includegraphics[width=0.3\textwidth]{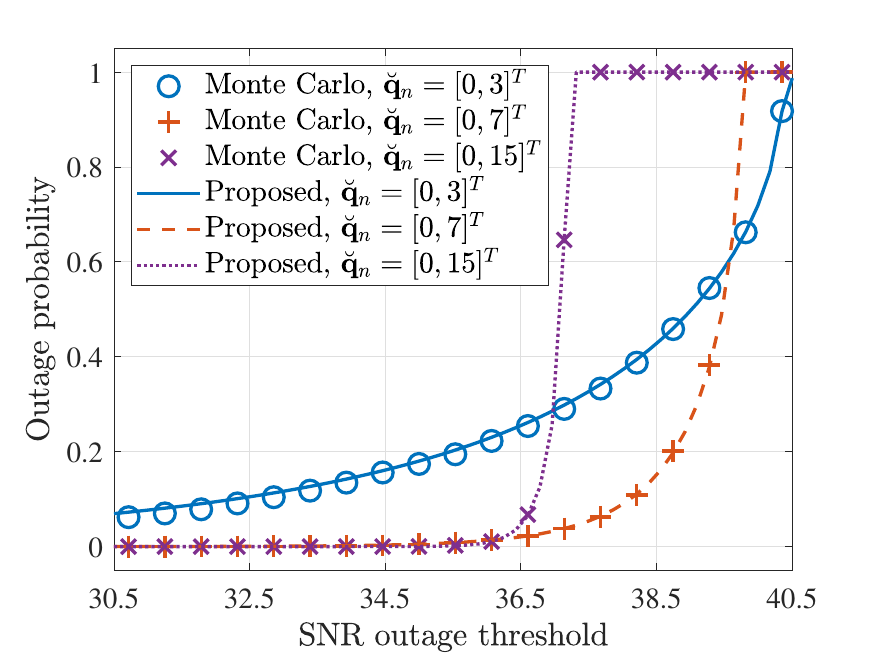}
    } 
    \subfigure[CORs at the prediction stage.]{
		\vspace{-1mm}
        \label{smgA:c}
		\includegraphics[width=0.3\textwidth]{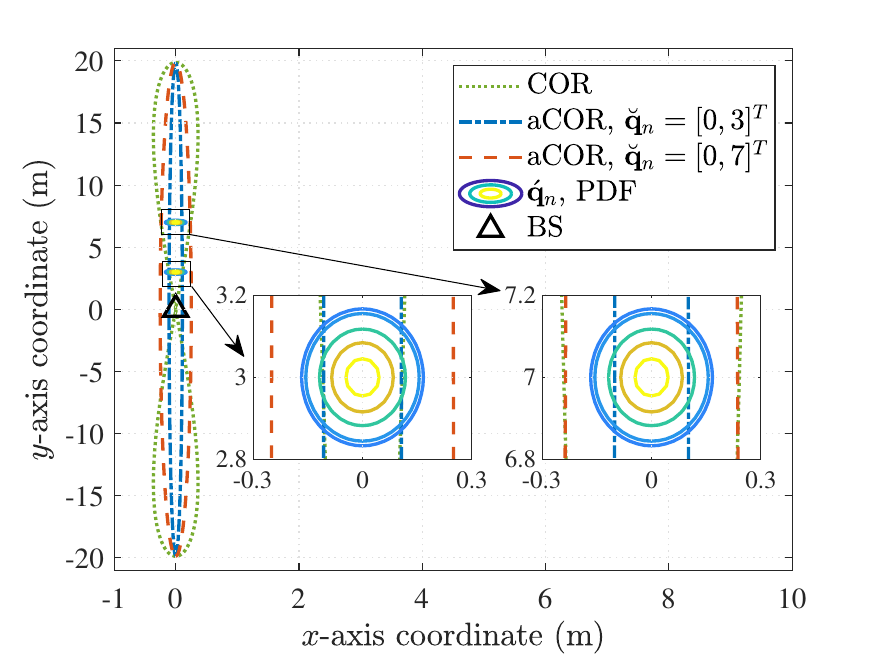}
    } 
     \label{smgA}
    \caption {Accuracies of proposed approximated OPs and CORs under different $\breve{\mathbf{q}}_{n}$. }
    \vspace{-3mm}
\end{figure*}

\begin{figure*}[!t]
	\centering
    \subfigure[OP versus $\breve{\mathbf{q}}_{n}$ with $N_{\text{t}} = 32$.]{
        \vspace{-1mm}
        \label{smgA:opq-32}
		\includegraphics[width = 0.8\columnwidth]{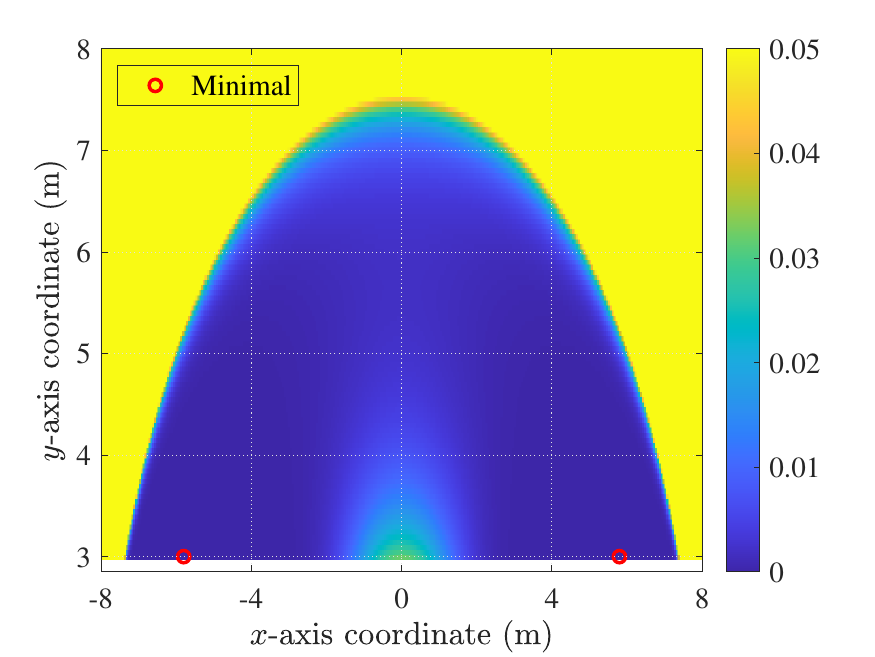}
	} 
    \subfigure[OP versus $\breve{\mathbf{q}}_{n}$ with $N_{\text{t}} = 64$.]{
		\vspace{-1mm}
        \label{smgA:opq-64}
		\includegraphics[width = 0.8\columnwidth]{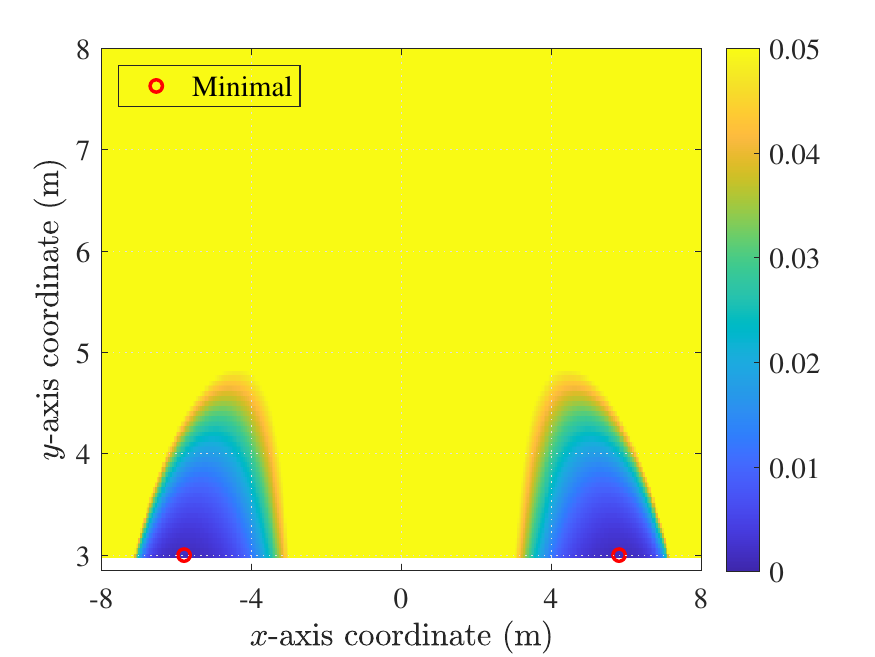}
    } 
    \subfigure[CORs with $N_{\text{t}} = 32$.]{
		\vspace{-1mm}
        \label{smgA:cor-32}
		\includegraphics[width = 0.8\columnwidth]{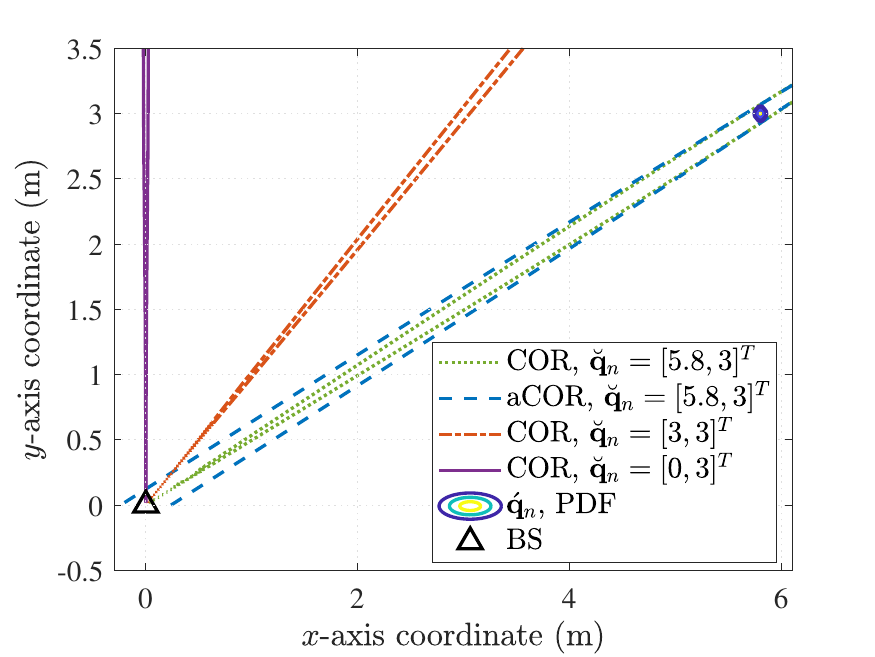}
    } 
    \subfigure[CORs with $N_{\text{t}} = 64$.]{
		\vspace{-1mm}
        \label{smgA:cor-64}
		\includegraphics[width = 0.8\columnwidth]{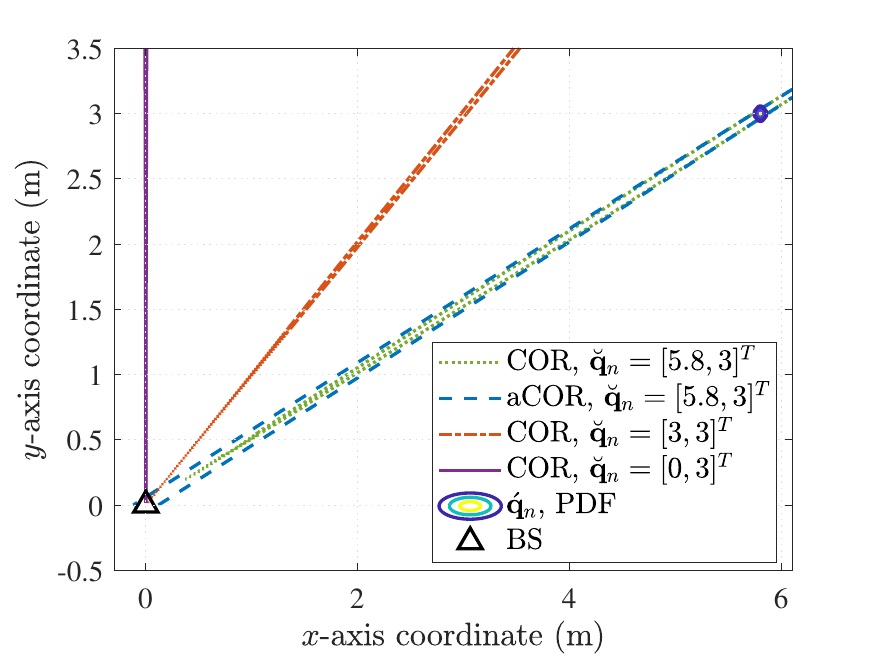}
    } 
     \label{smgA3}
    \caption {OPs and CORs with different $\breve{\mathbf{q}}_{n}$ and $N_{\text{t}}$. }
    \vspace{-3mm}
\end{figure*}

\subsection{Proposed OP Approximations}

Fig. \ref{smgA:a} and Fig. \ref{smgA:b} illustrate the accuracies of our proposed approximated OPs at the prediction and estimation stage and compare their differences under three representative predicted UAV positions.
Specifically, the Monte Carlo results in Fig. \ref{smgA:a} and Fig. \ref{smgA:b} are obtained by simulating the OP results with random noises (including the initial noise, process noise and measurement noise) in one time slot. 
The number of Monte Carlo simulation runs is set to $10^{4}$, and other specific system parameters are given by $a_{1} = a_{2} = 0.1$, $N_{\text{t}} = N_{\text{r}} = 16$, $N = 1$, $w_{n} = 0.5$ and $\mathbf{M}_{0} = 10^{-2}\mathbf{I}$ \cite{Relia2}, respectively. 
It can be observed that our proposed OP approximations closely match the Monte Carlo results in the cases with $\breve{\mathbf{q}}_{n} = [0, 7]^{T}$ and $\breve{\mathbf{q}}_{n} = [0, 15]^{T}$, thus validating the proposed approximation accuracy and effectiveness.\footnote{Unless specified otherwise, the OP refers to our proposed approximated OP in following paragraphs for brevity given the verified accuracy.} 
However, our proposed OP approximations are less accurate in the case with $\breve{\mathbf{q}}_{n} = [0, 3]^{T}$, especially at the prediction stage, which indicates that the proposed approximation accuracy is conditional on the UAV position. 
To explain such property, Fig. \ref{smgA:c} demonstrates the relationships among the dominant part of $\acute{\mathbf{q}}_{n}$ PDF, COR and aCOR in the cases with $\breve{\mathbf{q}}_{n} = [0, 3]^{T}$ and $\breve{\mathbf{q}}_{n} = [0, 7]^{T}$, respectively, given the target constant received SNR $\breve{\gamma}_{n} = 35$. 
Note that the COR with $\breve{\mathbf{q}}_{n} = [0, 3]^{T}$ is the same as that with $\breve{\mathbf{q}}_{n} = [0, 7]^{T}$, since (\ref{fm:cout-ineq}) is irrelevant to $\breve{x}_{n}$ and $\breve{y}_{n}$ given $\breve{x}_{n} = 0$.
In the scenario with $\breve{\mathbf{q}}_{n} = [0, 7]^{T}$, despite the seemingly considerable difference between the COR and aCOR, our proposed approximation is still accurate because both the COR and aCOR contain the dominant part of $\acute{\mathbf{q}}_{n}$ PDF, which verifies a condition for our proposed approximation being accurate: \emph{the prediction/estimation error must be sufficiently small such that the difference between COR and aCOR can have negligible impacts on the integral of the highly concentrated $\acute{\mathbf{q}}_{n}$ PDF}. 
In contrast, in the case with $\breve{\mathbf{q}}_{n} = [0, 3]^{T}$, both the COR and aCOR intersect with the dominant part of $\acute{\mathbf{q}}_{n}$ PDF, and thereby the difference between the COR and aCOR causes non-negligible approximation accuracy loss. 
Furthermore, although it is intractable to analytically characterize the relationship between the proposed approximation accuracy and the UAV position, (\ref{opt-cstrt-b}) is considered in this paper as a conservative but efficient constraint on UAV trajectories to avoid the low OP approximation accuracy, such as the case with $\breve{\mathbf{q}}_{n} = [0, 3]^{T}$. 

\begin{figure*}[!t]
	\centering
	\vspace{-4mm}
    \subfigure[Convergence behaviour of proposed algorithms.]{
        \vspace{-1mm}
        \label{smgB:cvg}
		\includegraphics[width = 0.8\columnwidth]{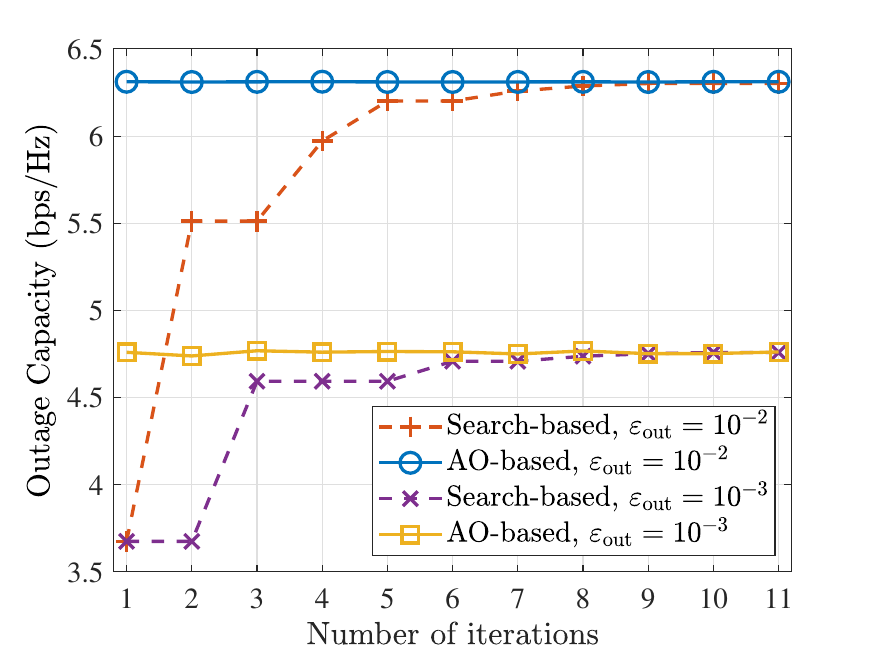}
	} 
    \subfigure[Outage capacity versus sensing duration ratio.]{
		\vspace{-1mm}
        \label{smgB:cw}
		\includegraphics[width = 0.8\columnwidth]{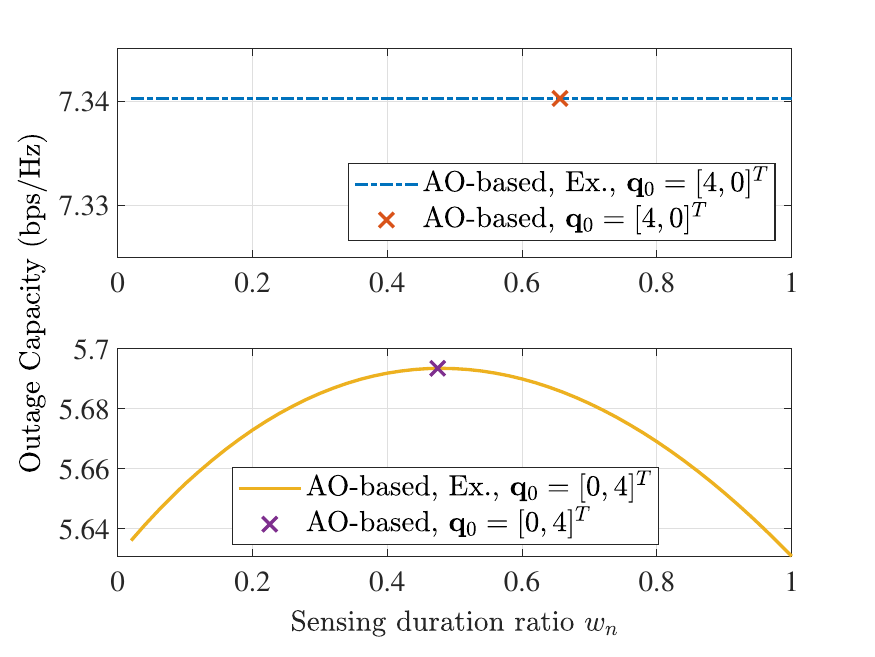}
    } 
    \caption {Performance of proposed algorithms.}
    \vspace{-3mm}
\end{figure*}

To obtain important insights into the relationship between the UAV trajectory and OPs, Fig. \ref{smgA:opq-32} and Fig. \ref{smgA:opq-64} illustrate the OP at the prediction stage within a given range of $\breve{\mathbf{q}}_{n}$ with $N_{\text{t}} = 32$ and $N_{\text{t}} = 64$, respectively.\footnote{The results about the approximated OP at the estimation stage is similar to those at the prediction stage, thus not presented for brevity.} 
A typically high target constant received SNR is set as $\breve{\gamma}_{n} = 0.975 \gamma_{\text{max}}$ for both cases, and other system parameters are specified as: $\mathbf{M}_{0} = 10^{-4}\mathbf{I}$, $N = 1$, $y_{\text{min}} = 3$ m and $a_{1} = a_{2} = 0.1$. 
As shown in Fig. \ref{smgA:opq-32} and Fig. \ref{smgA:opq-64}, the optimal predicted UAV trajectories resulting in the minimum OP exist at the line of $\breve{y}_{n} = y_{\text{min}}$, i.e., the minimum distance from the BS, with a certain $x$-axis coordinate given by $\pm 5.8$ m in both cases. 
Around the optimal predicted UAV trajectories, there exist certain regions where the OP is relatively low. 
Compared to the case with $N_{\text{t}} = 32$, the low-OP region with $N_{\text{t}} = 64$ becomes smaller and more concentrated at the optimal predicted UAV trajectories. 
Also, the positions near the direction $\breve{\theta}_{n} = 0^{\circ}$ are not contained in the low-OP region with $N_{\text{t}} = 64$. 
To explain such results, Fig. \ref{smgA:cor-32} and Fig. \ref{smgA:cor-64} show the accurate CORs with different predicted UAV trajectories $\breve{\mathbf{q}}_{n}$ corresponding to the cases in Fig. \ref{smgA:opq-32} and Fig. \ref{smgA:opq-64}, respectively. 
In both cases, the COR width increases when the predicted UAV trajectory $\breve{\mathbf{q}}_{n}$ varies from the direction $\breve{\theta}_{n} = 90^{\circ}$ to $\breve{\theta}_{n} = 0^{\circ}$, which is the main reason why the optimal predicted UAV trajectory $\breve{\mathbf{q}}_{n}^{*}$ is located at the line of $\breve{y}_{n} = y_{\text{min}}$. 
However, the UAV should be sufficiently close to the BS due to the potentially severe path loss, and the requirement of letting its dominant part of $\acute{\mathbf{q}}_{n}$ or $\grave{\mathbf{q}}_{n}$ PDF be contained in the COR. 
Therefore, \emph{the predicted UAV trajectory $\breve{\mathbf{q}}_{n}$ achieves a trade-off between minimizing the path loss and being covered by the mainlobe beam for minimizing the OP.}
Moreover, as illustrated in Fig. \ref{smgA:opq-32} and Fig. \ref{smgA:opq-64}, the smaller low-OP region with $N_{\text{t}} = 64$ is due to the narrower beam pattern generated by the larger transmit antenna number.

\subsection{Proposed Algorithms}

Fig. \ref{smgB:cvg} shows the convergence behaviour of our proposed search-based algorithm and AO-based algorithm in cases with $\varepsilon_{\text{out}} = 10^{-2}$ and $\varepsilon_{\text{out}} = 10^{-3}$. 
The initial state of the UAV is given by $\mathbf{x}_{0} = [0, 0, 4, 0]^{T}$ and the initial estimated state variables are represented by $\hat{\mathbf{x}}_{0} = \mathbf{x}_{0} + \mathbf{z}_{0}$ with $\mathbf{z}_{0} = [0.083, -0.001, 0.037, 0.042]^{T}$.
The other system parameters are given by: $a_{1} = a_{2} = 0.7$, $N_{\text{t}} = N_{\text{r}} = 64$, $N = 1$, and $\mathbf{M}_{0} = 10^{-3}\mathbf{I}$.
As shown in Fig. \ref{smgB:cvg}, the convergence of our proposed search-based algorithm is verified and our proposed AO-based algorithm also exhibits satisfactory convergence performance in both cases. 
Particularly, despite the slight fluctuation of the maximized outage capacity owing to the heuristic update of $\breve{\mathbf{q}}_{n}$, the output of our proposed AO-based algorithm approaches the maximum outage capacity much faster than the search-based algorithm, which demonstrates its effectiveness and considerably reduced computational complexity.  
Besides, compared to the case with $\varepsilon_{\text{out}} = 10^{-2}$, the maximum outage capacity significantly decreases and its fluctuation under the AO-based algorithm is more obvious in the case with $\varepsilon_{\text{out}} = 10^{-3}$, indicating the difficulty of maintaining a large outage capacity with a stringent OP tolerance threshold $\varepsilon_{\text{out}}$.

Fig. \ref{smgB:cw} demonstrates the varying trends of the outage capacity w.r.t. the sensing duration ratio $w_{n}$ with different UAV positions. 
The initial UAV states in the two cases are given by $\mathbf{x}_{0} = [4, 0, 0, 0]^{T}$ and $\mathbf{x}_{0} = [0, 0, 4, 0]^{T}$, respectively, to emphasize the different varying trends of outage capacity w.r.t. $w_{n}$. 
The other system parameters are as those in Fig. \ref{smgB:cvg} except the OP threshold given by $\varepsilon_{\text{out}} = 10^{-2}$. 
It can be observed that the impact of $w_{n}$ on the outage capacity with $\mathbf{q}_{0} = [0,4]^{T}$ is much larger than that with $\mathbf{q}_{0} = [4,0]^{T}$. 
This is because, when the UAV is at $[0,4]^{T}$, the state measurement provides a highly accurate estimation of the UAV trajectory and thus $C_{\text{e},n}$ can be quite larger than $C_{\text{p},n}$. 
Under such circumstances, the sensing duration ratio $w_{n}$ achieves \textit{a fundamental trade-off between sensing and sensing-assisted communication}: when $w_{n}$ is too small, the matched-filtering gain is insufficient to obtain highly accurate sensing results and thus cannot significantly enhance the communication efficiency or reliability; however, when $w_{n}$ is exceedingly large, the duration of enjoying the highly accurate beam alignment from sensing becomes limited, which also leads to sub-optimal communication performance.
In contrast to the case with $\mathbf{q}_{0} = [0,4]^{T}$, the sensing gain is negligible when the UAV is at $[4,0]^{T}$ due to the almost infinite measurement noise variance of the azimuth angle, resulting in the minor effect of sensing duration ratio $w_{n}$ on the outage capacity. 
Therefore, when the UAV trajectory is infavorable to sensing, incorporating the measured results contributes little to the outage capacity enhancement and thus the overhead for real-time state measurement can be saved. 
Moreover, Fig. \ref{smgB:cw} verifies that a near-optimal solution can be obtained by the subalgorithm for solving (P3.1) of the AO-based algorithm.

\begin{figure*}[!t]
	\centering
    \subfigure[UAV trajectories in the PMD case.]{
        \vspace{-1mm}
        \label{smgC:xy-p}
		\includegraphics[width=0.8\columnwidth]{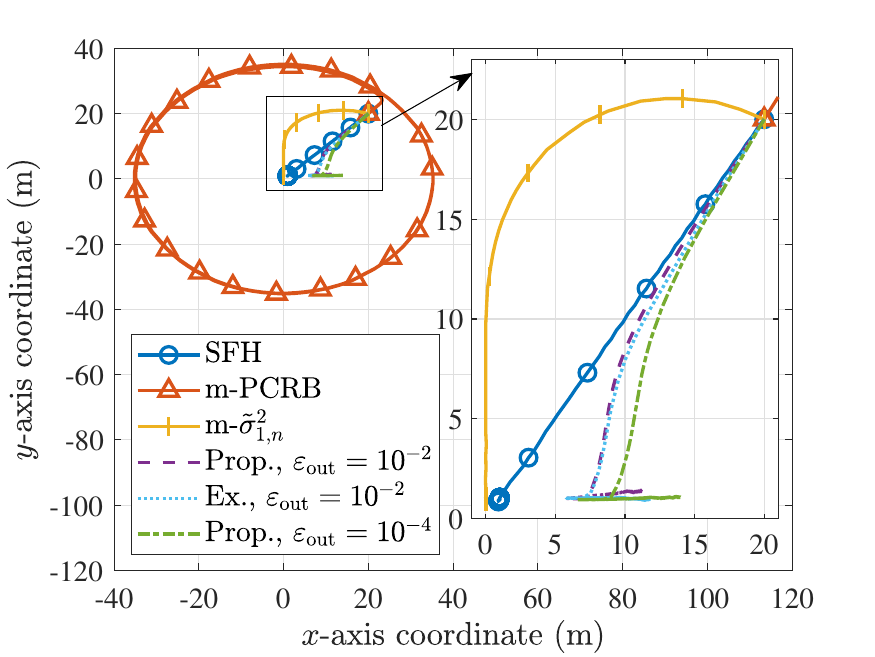}
	} 
    \subfigure[UAV trajectories in the PMnD case.]{
        \vspace{-1mm}
        \label{smgC:xy-m}
		\includegraphics[width=0.8\columnwidth]{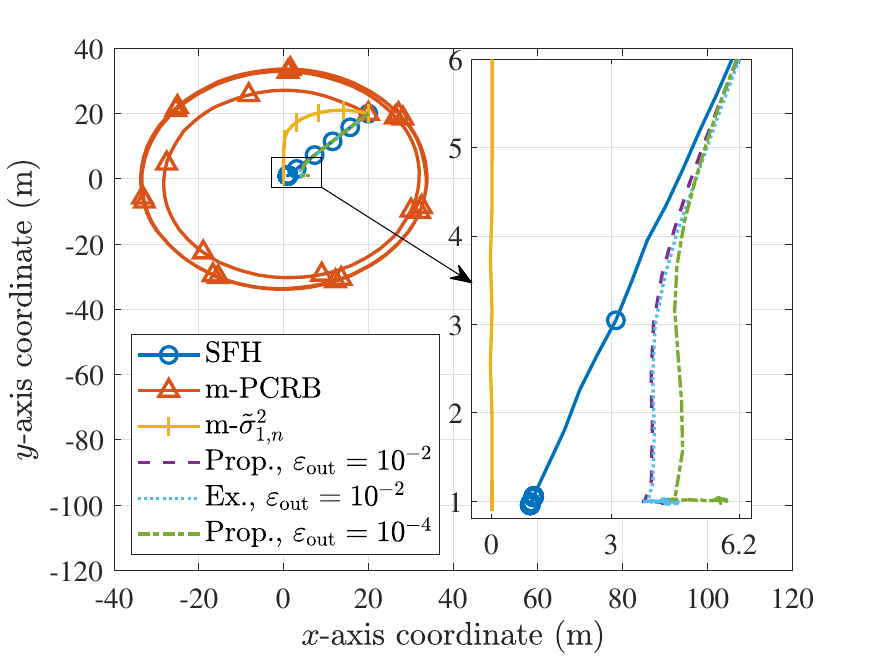}
	}  
    \caption {UAV trajectory comparisons between the PMD and PMnD case. }
    \vspace{-2mm}
\end{figure*}

\begin{figure*}[!t]
	\centering
    \subfigure[Outage capacities in the PMD case.]{
        \vspace{-1mm}
        \label{smgC:cn-p}
		\includegraphics[width=0.3\textwidth]{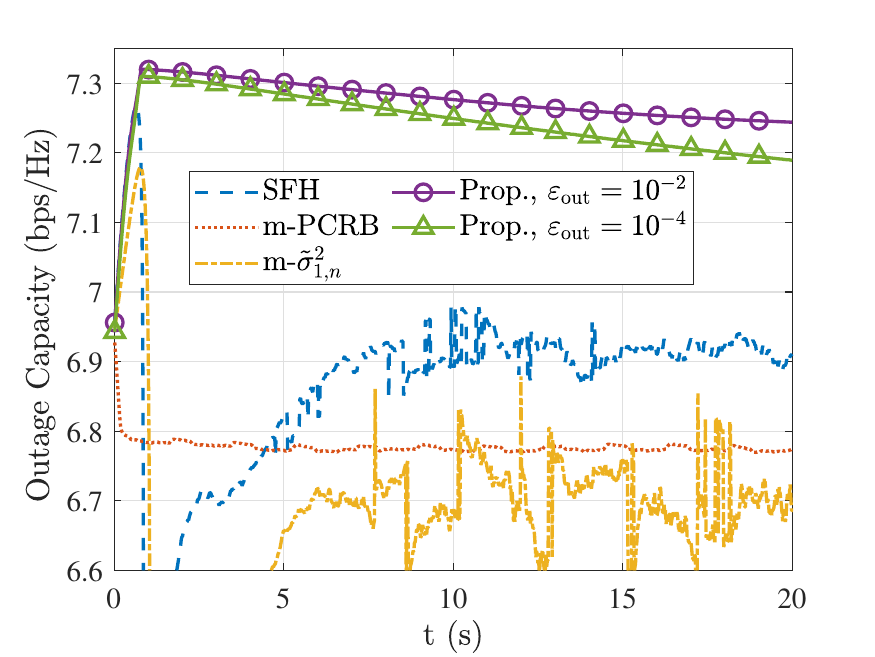}
	} 
    \subfigure[Outage capacities in the PMnD case.]{
        \vspace{-1mm}
        \label{smgC:cn-m}
		\includegraphics[width=0.3\textwidth]{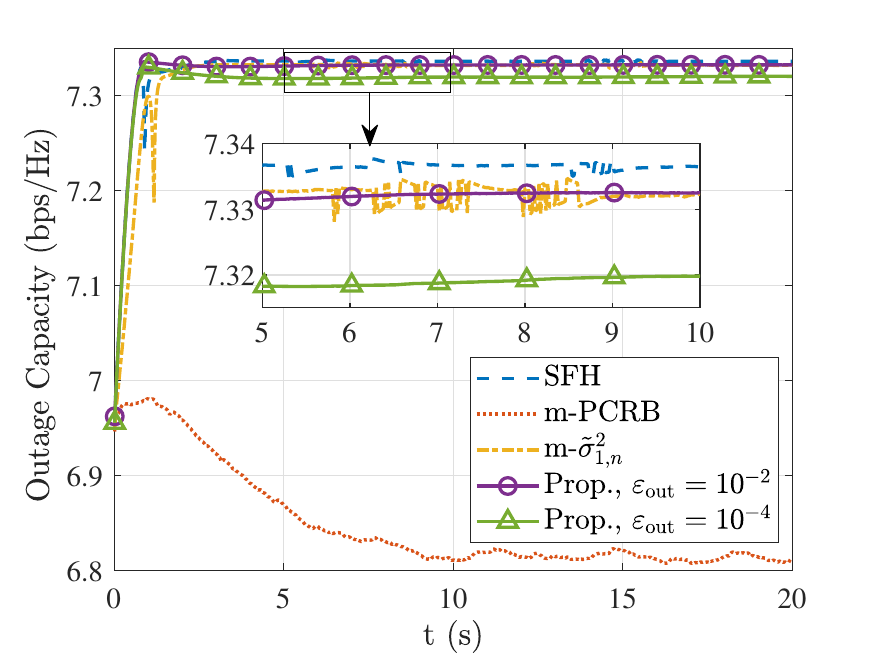}
	} 
    \subfigure[Sum of root PCRBs in the PMnD case.]{
        \vspace{-1mm}
        \label{smgC:s-m}
		\includegraphics[width=0.3\textwidth]{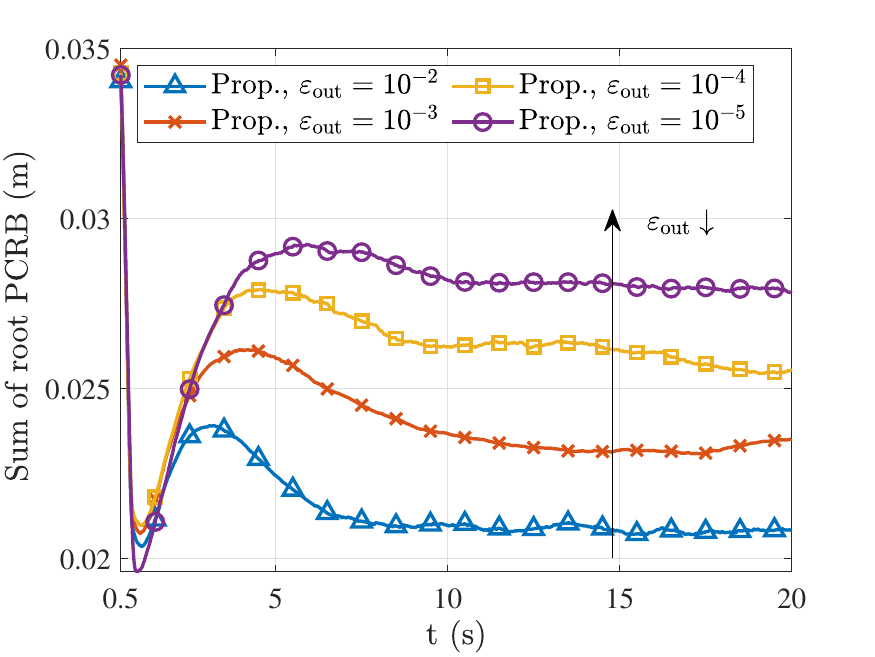}
	} 
    \caption {System performance comparsions among different schemes between the PMD and PMnD case.}
    \vspace{-3mm}
\end{figure*}

\vspace{-0.5mm}
\subsection{UAV Trajectories and System Performance} 
In this subsection, the results of our proposed UAV trajectory optimization scheme are compared with those of benchmarks in the prediction MSE-dominant (PMD) and prediction MSE-nondominant (PMnD) case, respectively. 
In the PMD case, the prediction MSE is so much smaller than the measurement MSE that $\mathbf{K}_{n} \approx \mathbf{0}$ holds \cite{welch1995}, which leads to the estimation MSE donimated by the prediction MSE, i.e., $\mathbf{M}_{n} \approx \mathbf{M}_{\text{p},n}$, due to (\ref{fm:efk-6}). 
Comparatively, the PMnD scenario refers to the case where the prediction MSE is not sufficiently smaller than the measurement MSE to satisfy $\mathbf{K}_{n} \approx \mathbf{0}$, indicating that the measurement MSE is small and the measured results are useful for decreasing the estimation MSE. 
In both cases, our proposed UAV trajectory design is compared with the following benchmarks:
\begin{itemize}
    \item Straight flight and hover (SFH): The UAV directly flies towards a specific position denoted by $\mathbf{q}_{\text{F}}$ with its maximum velocity $v_{\text{A,max}}$ and then hovers at $\mathbf{q}_{\text{F}}$ \cite{KTMeng-2023-TWC-UAV-IPSAC}.
    \item Posterior Cramér-Rao bound (PCRB) minimization (m-PCRB): At each time slot, the UAV trajectory is optimized to minimize the sum of predicted PCRBs for state variables of the next time slot, which can be expressed as \cite{fwd}
    \vspace{-1.5mm}
    \begin{equation}
        \min_{ \breve{\mathbf{q}}_{n} } \ \ \mathrm{Tr}(\mathbf{M}_{n}|_{\mathbf{x}_{n} = \breve{\mathbf{x}}_{n}}) \ \ 
        \text{s.t.} \ \  \text{(\ref{opt-cstrt-a})}. \notag 
        \vspace{-1.5mm}
    \end{equation}
    \item $\tilde{\sigma}_{1,n}^{2}$ minimization (m-$\tilde{\sigma}_{1,n}^{2}$): At each time slot, the UAV trajectory is optimized to minimize the approximated measurement noise variance for azimuth angle $\theta_{n}$ of the next time slot denoted by $\tilde{\sigma}_{1,n}^{2} = \sigma_{1,n}^{2}|_{\mathbf{x}_{n} = \breve{\mathbf{x}}_{n}}$, which can be expressed as
    \vspace{-1.5mm} 
    \begin{equation}
        \min_{ \breve{\mathbf{q}}_{n} } \ \ \tilde{\sigma}_{1,n}^{2} \ \ 
        \text{s.t.} \ \  \text{(\ref{opt-cstrt-a}), (\ref{opt-cstrt-b})}. \notag 
        \vspace{-1.5mm}
    \end{equation}
\end{itemize}
The sensing duration ratio $w_{n}$ is given by $w_{\text{max}} = 1$ to ensure the sensing performance as much as possible for all benchmarks. 

\subsubsection{UAV trajectories}
Fig. \ref{smgC:xy-p} and Fig. \ref{smgC:xy-m} illustrate the UAV trajectories obtained by the benchmarks and our proposed UAV trajectory optimization scheme in both the PMD and PMnD cases. 
To compare the dynamical UAV trajectories under different schemes during the whole $N\Delta T = 20$ s, the initial UAV motion state under all schemes are uniformly set as $\mathbf{x}_{0} = [ 20, 0, 20, 0 ]^{T}$.
The measurement capability coefficients are set as $a_{1} = a_{2} = 1$ and $a_{1} = a_{2} = 0.1$ for the PMD and PMnD case, respectively. 
To fairly compare our proposed scheme and benchmarks, the constraint (\ref{opt-cstrt-b}) with $y_{\text{min}} = 1$ is also applied in the m-$\tilde{\sigma}_{1,n}^{2}$ scheme and the specific position under the SFH scheme is given by $\mathbf{q}_{\text{F}} = [1,1]^{T}$. 
Other system parameters are given by $\tilde{q} = 10^{-5}$. 
First, it can be observed from both Fig. \ref{smgC:xy-p} and Fig. \ref{smgC:xy-m} that the UAV trajectory obtained by our proposed AO-based algorithm (dentoed by ``Prop.'') well match the results obtained by the exhaustive search (denoted by ``Ex.''), which validates the effectiveness of our proposed AO-based algorithm. 
Second, in both the PMD and PMnD cases, the UAV trajectory under the m-PCRB scheme is approximately circular to maintain an optimal distance minimizing the PCRB, 
while the UAV under the m-$\tilde{\sigma}_{1,n}^{2}$ scheme approaches the BS in the direction of $\theta_{n} = 90^{\circ}$ and then hovers around $[0,1]^{T}$, which is the optimal position for minimizing $\tilde{\sigma}_{1,n}^{2}$. 
Different from the UAV trajectories under benchmarks, the UAV under our proposed scheme tends to approach the BS with a relatively smaller azimuth angle and then stay at the straight line with $\breve{y}_{n} = y_{\text{min}}$ parallel to the BS ULA antennas in both the PMD and PMnD cases. 
The reason for such trajectory is that being at the line with $\breve{y}_{n} = y_{\text{min}}$ leads to wide COR/beam coverage, which is consistent with our previous observation from Fig. \ref{smgA:opq-32} and Fig. \ref{smgA:opq-64} and also demonstrates the importance of beam coverage to signal reception reliability. 
\emph{Consequently, the predicted UAV trajectory parallel to the BS ULA antennas with $\breve{y}_{n} = y_{\text{min}}$ is advantageous for outage capacity maximization.} 
In addition, the UAV trajectory with $\varepsilon_{\text{out}} = 10^{-4}$ is generally farther away from the BS than that with $\varepsilon_{\text{out}} = 10^{-2}$ in both the PMD and PMnD cases, indicating that a larger UAV-BS distance is more beneficial for enhancing the communication reliability.

\subsubsection{Outage capacities}
Fig. \ref{smgC:cn-p} and Fig. \ref{smgC:cn-m} compare the outage capacities achieved by the benchmarks and our proposed scheme in both cases.
Particularly, the outage capacities of benchmarks are calculated by our proposed algorithm for (P3.1) given their optimized predicted UAV trajectories and the OP threshold $\varepsilon_{\text{out}} = 10^{-2}$.
As illustrated in Fig. \ref{smgC:cn-p}, the communication performances under the SFH and m-$\tilde{\sigma}_{1,n}^{2}$ scheme exhibit large random variations similar as fast fadings in the PMD case. 
The reason is that the UAV is improperly near the BS and can be easily away from the COR/beam coverage due to the position uncertainty. 
The outage capacity under the m-PCRB scheme are relatively stable but limited by the high path loss. 
Comparatively, the outage capacity under our proposed scheme is much more stable than benchmarks and also higher than benchmarks for over 0.2 bps/Hz, which validates the effectiveness and superiorities of our proposed outage capacity maximization scheme over benchmarks in the PMD case. 
Nevertheless, Fig. \ref{smgC:cn-m} shows that such superiorities disappear in the PMnD case because the small measurement MSE leads to a low OP even if the UAV is close to the BS. 
Besides, the outage capacity under our proposed scheme in the case with $\varepsilon_{\text{out}} = 10^{-4}$ is lower than that with $\varepsilon_{\text{out}} = 10^{-2}$, which shows the trade-off between the communication reliability and efficiency. 

\subsubsection{Sensing accuracies}
Fig.\ref{smgC:s-m} compares the sensing accuracies under our proposed scheme among different $\varepsilon_{\rm{out}}$ in the PMnD case.
The sensing accuracy is characterized by the sum of root PCRBs for $x_{n}$ and $y_{n}$ given by $\sqrt{ [\mathbf{M}_{n}]_{11} } + \sqrt{ [\mathbf{M}_{n}]_{33} }$ \cite{fwd}.
It can be seen that the sensing accuracy gradually maintains stable with the increasing of time resulted from the little variance of UAV trajectory.
Besides, the sensing accuracy is generally higher with a larger $\varepsilon_{\rm{out}}$.
This is because the UAV tends to obtain a smaller azimuth angle w.r.t. the BS for a wider beam coverage, by which means the more stringent OP tolerance constraint can be satisfied.
However, such UAV trajectory can degrade the measurement MSE for the azimuth angle, which further results in the larger estimation MSE.
Therefore, the UAV trajectory also achieves a trade-off between the sensing accuracy and the tolerated minimum OP in our system.

\section{Conclusions}

This paper studied the outage capacity maximization for UAV tracking enabled by sensing-assisted predictive beamforming, where the UAV trajectory, sensing duration ratio, and target constant received SNRs were jointly optimzied. 
To facilitate the formulation of a tractable optimization problem, closed-form OP approximations were proposed based on second-order Taylor expansions, which also characterized the outage capacity. 
Then, two efficient algorithms were proposed to address the non-convex approximated optimization problem: a search-based algorithm with ensured convergence and an AO-based algorithm with lower complexity. 
Simulation results verified the effectiveness of our proposed approximations, algorithms, and the superiority of the proposed joint UAV tracking and outage capacity maximization scheme over benchmarks in the PMD case. 
Furthermore, our results demonstrated that the optimal predicted UAV trajectory tended to be parallel to the BS ULA antennas with a nonzero minimum distance, achieving a trade-off between decreasing path loss and increasing beam coverage area for outage capacity maximization.
The extension of our proposed approximations to multi-static ISAC systems are worthwhile future works. 

\section*{Appendix A \\ \textsc{Proof of Proposition 1}}

According to the EKF framework \cite{MKay}, the state vector at each time slot can be approximately Gaussian distributed, represented by $\mathbf{x}_{n-1} \sim \mathcal{N}(\hat{\mathbf{x}}_{n-1},\mathbf{M}_{n-1}), \forall n\in\{1,2,...,N\}$.  
Thus, $\mathbf{x}_{n} \sim \mathcal{N}(\breve{\mathbf{x}}_{n}, \mathbf{M}_{\text{p},n})$ is derived from (\ref{fm:pre}) and (\ref{fm:Mpn}).  
Furthermore, as a marginal distribution of the state vector $\mathbf{x}_{n}$, the ground-truth UAV trajectory $\mathbf{q}_{n}$ is also Gaussian distributed given by $\mathbf{q}_{n} \sim \mathcal{N}(\breve{\mathbf{q}}_{n},\bm{\breve{\Lambda}}_{n})$ with 
\vspace{-1.5mm}
\begin{equation}
    \bm{\breve{\Lambda}}_{n} = 
    \begin{bmatrix}
        \breve{\Lambda}_{\text{x},n}^{2} & \breve{\Lambda}_{\text{xy},n}^{2} \\
        \breve{\Lambda}_{\text{xy},n}^{2} & \breve{\Lambda}_{\text{y},n}^{2}
    \end{bmatrix} = 
    \begin{bmatrix}
        [\mathbf{M}_{\text{p},n}]_{11} & [\mathbf{M}_{\text{p},n}]_{13} \\
        [\mathbf{M}_{\text{p},n}]_{31} & [\mathbf{M}_{\text{p},n}]_{33}
    \end{bmatrix}. 
    \vspace{-1.5mm}
\end{equation}
Therefore, $\acute{\mathbf{q}}_{n} \sim \mathcal{N}(\mathbf{0},\bm{\breve{\Lambda}}_{n})$ holds. 
Note that $\tilde{\xi}_{\text{p},n}$ is a univariate function of $\acute{\mathbf{q}}_{n}$. 
Then, the approximated OP at the prediction stage of the $n$th time slot (\ref{fm:aCOP}) can be derived from (\ref{fm:OPexpress}) with $\mathcal{Q}_{\text{p},n} \approx \tilde{\mathcal{Q}}_{\text{p},n}$, i.e., 
\vspace{-1.5mm}
\begin{equation}
    \zeta_{\text{p},n} 
    \approx \tilde{\zeta}_{\text{p},n} = 1 - \int_{\acute{x}_{\text{L}}}^{\acute{x}_{\text{U}}} 
    \biggl( \int_{\acute{y}_{\text{L}}(\acute{x}_{n})}^{\acute{y}_{\text{U}}(\acute{x}_{n})}
    f(\acute{\mathbf{q}}_{n}) \mathrm{d}\acute{y}_{n} \biggr) \mathrm{d}\acute{x}_{n}, 
    \vspace{-1.5mm}
\end{equation}
where $f(\cdot)$ denotes the Gaussian PDF of $\acute{\mathbf{q}}_{n}$, $\acute{x}_{\text{L}}$ and $\acute{x}_{\text{U}}$ can be obtained from the equation $\mathrm{d}\acute{x}_{n} / \mathrm{d}\acute{y}_{n} = 0 $. 
This completes the proof.

\section*{Appendix B \\ \textsc{Proof of Proposition 2}}

For notational simplicity, we adopt $\chi$ to represent either the function $\breve{\chi}$ or $\hat{\chi}$. 
Accordingly, $\gamma, \bar{x}, \chi_{\text{U}}, \chi_{\text{L}}, \Lambda_{\text{x},n}, \mathbf{\Lambda}_{n}$ represents $\breve{\gamma}_{n}, \acute{x}_{n}, \breve{\chi}_{\text{U}}, \breve{\chi}_{\text{L}}, \breve{\Lambda}_{\text{x},n}, \breve{\mathbf{\Lambda}}_{n} $ in the case where $\chi$ denotes $\breve{\chi}$, and $\hat{\gamma}_{n}, \grave{x}_{n} \hat{\chi}_{\text{U}}, \hat{\chi}_{\text{L}}, \hat{\Lambda}_{\text{x},n}, \hat{\mathbf{\Lambda}}_{n}$ in the case where $\chi$ denotes $\hat{\chi}$, respectively. 
Then, an upperbound of the partial derivative of $\chi$ w.r.t. $\gamma$ can be derived as 
\vspace{-1.5mm}
\begin{equation}
    \frac{\partial\chi}{\partial\gamma} \leq A \biggl( \frac{\partial Y_{1}}{\partial \gamma} + (\bar{x} + \breve{x}_{n})^{2} \frac{\partial Y_{2}}{\partial \gamma} \biggr) = A \rho(\gamma), 
    \vspace{-1.5mm}
\end{equation}
with $A = \frac{\Lambda_{\text{x},n} \max\{ \frac{e^{-\chi_{\text{U}}(\gamma)^{2}}}{\sqrt{\pi}}, \frac{e^{-\chi_{\text{L}}(\gamma)^{2}}}{\sqrt{\pi}} \} }{\sqrt{2|\mathrm{det}(\mathbf{\Lambda}_{n})|(Y_{1}+Y_{2}(\bar{x} + \breve{x}_{n})^{2})}} $.
The derivative of $\rho(\gamma)$ w.r.t. $\gamma$ can be derived as $\frac{\mathrm{d}\rho}{\mathrm{d}\gamma} = \frac{(\breve{x}_{n}^{2} + \breve{y}_{n}^{2})^{6} \tilde{P}(\rho_{1} \gamma + \rho_{0})}{((\breve{x}_{n}^{2} + \breve{y}_{n}^{2})^{3}\gamma + \breve{x}_{n}^{2}\breve{y}_{n}^{2}\tilde{P}M)^{4}}$,
where the specific expressions of $\rho_{0}$ and $\rho_{1}$ are given by 
\begin{align}
    \rho_{0} &= 2 \breve{x}_{n}^{2} \breve{y}_{n}^{2} \tilde{P} M ( \breve{y}_{n}^{2}( ( \acute{x}_{n} + \breve{x}_{n})^{2}( \breve{x}_{n}^{2} + 2\breve{y}_{n}^{2} ) \notag \\
    &+ \breve{x}_{n}^{2} H^{2} )M + (\breve{x}_{n}^{2} + \breve{y}_{n}^{2})^{3}N_{\text{t}} ) \geq 0, \\
    \rho_{1} &= 2 \breve{y}_{n}^{2} (\breve{x}_{n}^{2} + \breve{y}_{n}^{2})^{3} ((\acute{x}_{n} + \breve{x}_{n})^{2}(\breve{x}_{n}^{2} - \breve{y}_{n}^{2}) + \breve{x}_{n}^{2} H^{2} ) M \notag \\
    &+ 2 (\breve{x}_{n}^{2} + \breve{y}_{n}^{2})^{6} N_{\text{t}}. 
\end{align}
Next, two cases with $\rho_{1} \geq 0 $ and $\rho_{1} < 0 $ are discussed, respectively. 
For the case with $\rho_{1} \geq 0 $, $\frac{\mathrm{d}\rho}{\mathrm{d}\gamma} \geq 0$ holds due to $ \gamma \geq 0 $.
As for the case with $\rho_{1} < 0 $, $\frac{\mathrm{d}\rho}{\mathrm{d}\gamma}$ is a monotonically nonincreasing function. 
Since both $\frac{\mathrm{d}\rho}{\mathrm{d}\gamma}|_{\gamma = 0} \geq 0$ and $\lim\limits_{\gamma \rightarrow \infty} \frac{\mathrm{d}\rho}{\mathrm{d}\gamma} \geq 0 $ can be obtained, $\frac{\mathrm{d}\rho}{\mathrm{d}\gamma} \geq 0$ also holds in this case. 
Thus, $\rho(\gamma)$ is a monotonically nondecreasing function of $\gamma$. 
Finally, both $\rho(0) \leq 0$ and $\lim\limits_{\gamma \rightarrow \infty} \rho(\gamma) = 0$ can be obtained. 
As a result, $\frac{\partial\chi}{\partial\gamma} \leq A \rho(\gamma) \leq 0$ holds, completing the proof. 

\bibliographystyle{IEEEtran}
\bibliography{IEEEabrv,ref}

\end{document}